\documentclass[11pt]{article}
\usepackage[left=2cm,right=2cm, top=2.5cm,bottom=2.5cm,bindingoffset=0cm]{geometry}
\usepackage{tikz, titlesec, tikz-cd}
\usetikzlibrary{calc,intersections,through,backgrounds}
\usepackage{amsmath, amssymb}
\usepackage{amsthm}
\usepackage{mathtools}
\usepackage{booktabs, makecell}
\usepackage{adjustbox}

\mathtoolsset{showonlyrefs}

\usepackage{array, hyperref}

\titleformat{\section}{\large\bfseries\filcenter}{\thesection}{1em}{}
\titleformat{\subsection}{\bfseries}{\thesubsection}{1em}{}
\titleformat{\subsubsection}[runin]{\bfseries}{\thesubsubsection}{1em}{}[.]

\usepackage{multirow}

\usetikzlibrary{positioning}
\usetikzlibrary{decorations.text}
\usetikzlibrary{decorations.pathmorphing}
\usetikzlibrary{decorations.markings}

\usetikzlibrary{arrows} 
\tikzset{
    >=stealth',
    pil/.style={
           ->,
           ,
           shorten <=2pt,
           shorten >=2pt,}
}

\newcommand{\Z}{\mathbb{Z}}

\newcommand{\R}{\mathbb{R}}
\renewcommand{\P}{\mathcal{P}}

\newcommand{\N}{\mathbb{\Z}_{\geq 0}}
\newcommand{\Id}{\mathrm{Id}}

\newcommand{\RP}{\mathbb{R}\mathrm{P}}
\newcommand{\GL}{\mathrm{GL}}
\newcommand{\PGL}{\mathrm{PGL}}
\newcommand{\D}{\mathcal D}

\newcommand{\Sh}{T}

\newcommand{\gl}{\Mat}

\newcommand{\Ker}[1]{\mathrm{Ker}\,#1}
\newcommand{\g}{\mathfrak{G}}
\newcommand{\ord}[1]{\mathrm{ord}\,#1}
\newcommand{\sgn}{\mathrm{sgn}}
\newcommand{\Mat}{\mathrm{Mat}}
\newcommand{\ttimes}{\,\tilde \times \,}
\newcommand{\Lax}{\mathcal{L}}
\newcommand{\DO}[2]{\mathrm{DO}_{#1}(#2)}

\newcommand{\ALLDO}[1]{\mathrm{DO}_{#1}}
\newcommand{\IDO}[2]{\mathrm{IDO}_{#1}(#2)}
\newcommand{\PBDO}[2]{\mathrm{PBDO}_{#1}(#2)}
\newcommand{\PSIDO}[1]{\Psi\mathrm{DO}_{#1}}
\newcommand{\IPSIDO}[1]{\mathrm{I}\Psi\mathrm{DO}_{#1}}
\newcommand{\Tr}[1]{\mathrm{Tr}\,#1}

\newcommand{\CPM}[2]{  P(#1, #2)}

\newcommand{\Ad}{{\mathrm{Ad}}}
\newcommand{\grad}{{\mathrm{grad}\,}}
\newcommand{\Proj}{\mathbb P}
\newcommand{\pos}{{>0}}
\renewcommand{\neg}{{<0}}

\newtheorem{lemma}{Lemma}[section]
\newtheorem{proposition}[lemma]{Proposition}
\newtheorem{theorem}[lemma]{Theorem}
\newtheorem{corollary}[lemma]{Corollary}

\theoremstyle{definition}
\newtheorem{remark}[lemma]{Remark}
\newtheorem{definition}[lemma]{Definition}

\newtheorem{example}[lemma]{Example}

\makeatletter
\renewenvironment{proof}[1][\proofname] {\par\pushQED{\qed}\normalfont\topsep6\p@\@plus6\p@\relax\trivlist\item[\hskip\labelsep\bfseries#1\@addpunct{.}]\ignorespaces}{\popQED\endtrivlist\@endpefalse}
\makeatother

\usetikzlibrary{arrows} 
\tikzset{
    >=stealth',
    pil/.style={
           ->,
           ,
           shorten <=2pt,
           shorten >=2pt,}
}
\tikzset{->-/.style={decoration={
  markings,
  mark=at position .7 with {\arrow{>}}},postaction={decorate}}}
  \tikzset{a/.style={decoration={
  markings,
  mark=at position .52 with {\arrow{angle 90}}},postaction={decorate}}}
\tikzset{-<-/.style={decoration={
  markings,
  mark=at position .4 with {\arrow{<}}},postaction={decorate}}}  

\title{Pentagram maps and refactorization in Poisson-Lie groups} 
\author{Anton Izosimov\thanks{Department of Mathematics, University of Arizona, e-mail: \tt{izosimov@math.arizona.edu}}}
\date{}

\begin{document}
\maketitle
\vspace*{-0.7cm}
\abstract{
The pentagram map was introduced by R.~Schwartz in 1992 and is now one
of the most renowned discrete integrable systems. In the present paper
we prove that this  map, as well as all its known integrable
multidimensional generalizations, can be seen as refactorization-type
mappings in the Poisson-Lie group of pseudo-difference operators.
This brings the pentagram map into the rich framework of Poisson-Lie
groups, both describing new structures and simplifying and revealing
the origin of its known properties. In particular, for multidimensional pentagram maps the
Poisson-Lie group setting provides new Lax forms with a spectral
parameter and, more importantly, invariant Poisson structures in all
dimensions, the existence of which has been an open problem since the
introduction of those maps. Furthermore, for the classical pentagram map our approach naturally yields its combinatorial description in terms of weighted directed networks and cluster algebras.
}

\tableofcontents

\section{Introduction and outline of main results}

The pentagram map, introduced by R.\,Schwartz in \cite{schwartz1992pentagram}, is a discrete integrable system on the space of projective equivalence classes of planar polygons. 
The definition of this map is illustrated in Figure~1: the image of the polygon $P$ under the pentagram map is the polygon $P'$ whose vertices are the intersection points of consecutive shortest diagonals of~$P$ (i.e.,  diagonals connecting second-nearest vertices). The pentagram map has been an especially popular subject in the last decade, mainly due to a combination of an elegant geometric definition and connections to such topics as cluster algebras, dimer models etc.
\begin{figure}[h]
\centering
\begin{tikzpicture}[]
\coordinate (VK7) at (0,0);
\coordinate (VK6) at (1.5,-0.5);
\coordinate (VK5) at (3,1);
\coordinate (VK4) at (3,2);
\coordinate (VK3) at (1,3);
\coordinate (VK2) at (-0.5,2.5);
\coordinate (VK1) at (-1,1.5);

\draw  [line width=0.5mm]  (VK7) -- (VK6) -- (VK5) -- (VK4) -- (VK3) -- (VK2) -- (VK1) -- cycle;
\draw [dashed, line width=0.2mm, name path=AC] (VK7) -- (VK5);
\draw [dashed,line width=0.2mm, name path=BD] (VK6) -- (VK4);
\draw [dashed,line width=0.2mm, name path=CE] (VK5) -- (VK3);
\draw [dashed,line width=0.2mm, name path=DF] (VK4) -- (VK2);
\draw [dashed,line width=0.2mm, name path=EG] (VK3) -- (VK1);
\draw [dashed,line width=0.2mm, name path=FA] (VK2) -- (VK7);
\draw [dashed,line width=0.2mm, name path=GB] (VK1) -- (VK6);

\path [name intersections={of=AC and BD,by=Bp}];
\path [name intersections={of=BD and CE,by=Cp}];
\path [name intersections={of=CE and DF,by=Dp}];
\path [name intersections={of=DF and EG,by=Ep}];
\path [name intersections={of=EG and FA,by=Fp}];
\path [name intersections={of=FA and GB,by=Gp}];
\path [name intersections={of=GB and AC,by=Ap}];

\draw  [line width=0.5mm]  (Ap) -- (Bp) -- (Cp) -- (Dp) -- (Ep) -- (Fp) -- (Gp) -- cycle;

\node at (-0.9,2.3) () {$P$};
\node at (2,1.5) () {$P'$};

\end{tikzpicture}
\caption{The pentagram map.}\label{Fig1}
\end{figure}

Integrability of the pentagram map was established, in different contexts, in \cite{ovsienko2010pentagram, ovsienko2013liouville, soloviev2013integrability}.  Furthermore, it was shown that the pentagram map can be viewed as a particular case of several general constructions of integrable systems. In particular, it has an interpretation in terms of cluster algebras~\cite{GLICK20111019}, networks of surfaces \cite{Gekhtman2016},  T-systems \cite{kedem2015t}, and Poisson-Lie groups \cite{fock2014loop}. In the present paper we suggest an alternative to \cite{fock2014loop} Poisson-Lie approach to the pentagram map. Namely, we show that the pentagram map can be seen as a refactorization in the Poisson-Lie group of pseudo-difference operators. The main advantage of our approach is that it is based on the geometric
definition of the map and the explicit formulas are obtained as its corollaries. We thereby obtain all the ingredients needed to establish integrability, namely an invariant Poisson structure, Lax representation, and first integrals, directly from geometry.  This can be compared with other frameworks, in particular, the ones based on cluster
algebras \cite{Gekhtman2016} and Poisson-Lie groups  \cite{fock2014loop}, which lead to integrable maps shown to coincide with the
pentagram map at the level of formulas.


By virtue of the geometric nature of our approach, it almost immediately generalizes to pentagram-type maps in higher dimensions and enables us to treat all these maps on an equal footing. It turns out that our scheme covers all previously known higher-dimensional integrable cases, and also gives rise to a large number of new ones. Furthermore, for many of the previously known integrable maps our approach provides certain missing ingredients, in particular invariant Poisson structures for \textit{short-diagonal} and \textit{dented} maps of~\cite{khesin2013, khesin2016}. Construction of such structures has been an open problem since the introduction of these maps. Furthermore, for these maps we get new Lax representations which are, in a sense, dual to the ones given in~\cite{khesin2013, khesin2016}. 


Recall that a \textit{refactorization} is a mapping of the form $AB \mapsto BA$, where $A$ and $B$ are elements of a non-Abelian group, e.g. matrices. 
The relation between such mappings and integrability was pointed out in
 \cite{veselov1991integrable, moser1991discrete} and put in the context of Poisson-Lie groups in \cite{deift1991poisson}. Nowadays, refactorization in Poisson-Lie groups is viewed as one of the most universal mechanisms of integrability for discrete dynamical systems. 
 In this paper we suggest such an interpretation for the pentagram map and its generalizations. Below we briefly describe the construction for the case of the classical pentagram map. \par

Let $\{ v_i \in \RP^2 \}$ be a planar $n$-gon, and let $\{V_i \in \R^3\}$ be its arbitrary lift to $\R^3$ (here and in what follows we assume that the ground field is real numbers, although all the same constructions work over~$\mathbb C$). The sequence $V_i$ can be encoded by writing down the relations between quadruples of consecutive vectors:
$$
a_i V_i + b_i V_{i+1} + c_i V_{i+2} + d_i V_{i+3} = 0,
$$
where $i \in \Z$, and $a, b, c, d$ are $n$-periodic sequences of real numbers. This can be equivalently written as
$
\D V = 0,
$
where $V$ is a bi-infinite sequence whose entries are the vectors $V_i$, and $\D$ is an $n$-periodic {difference operator}
$$
\D := a+ b \Sh + c \Sh^2 + d \Sh^3.
$$
Here $\Sh$ is the left shift operator  on bi-infinite sequences, $(TV)_i := V_{i+1}$, while the sequences $a,b,c,d$  of real numbers act on sequences of vectors by term-wise multiplication: $(aV)_i := a_i V_i$.
Thus, one can encode planar polygons by \textit{third order difference operators}. There is, however, more than one operator corresponding to a given polygon in $\RP^2$. Namely, one can multiply~$\D$ by scalar sequences from the left or right without changing the corresponding polygon. This means that, for any mapping of the space of polygons to itself, its lift to difference operators is not a map, but a correspondence (a multivalued map). 
To explicitly describe this correspondence for the case of the pentagram map, we split the difference operator $\D = a+ b \Sh + c \Sh^2 + d \Sh^3$ into two parts:
$$
\D_+ := a+ c \Sh^2, \quad
\D_- := b \Sh + d \Sh^3.
$$
\begin{theorem}\label{thm0}
The pentagram map, written in terms of difference operators, is a multivalued map $$\D =  \D_+ + \D_- \quad \longmapsto \quad \tilde \D =  \tilde \D_+ + \tilde \D_-$$ determined by the relation
\begin{equation}\label{mainRelation}
\tilde{\D}_+  {\D}_- = \tilde {\D}_-  {\D}_+.
\end{equation}
\end{theorem}
\begin{proof}[Proof sketch]
Equation \eqref{mainRelation} can be viewed as a homogeneous linear system on $4n$ unknown coefficients of the $n$-periodic operator $\tilde \D$. Both sides of  \eqref{mainRelation} are linear combinations of $\Sh$,  $\Sh^3$, and $\Sh^5$ with $n$-periodic coefficients, so the number of equations is $3n$, which is less than the number of unknowns. Therefore, there always exists a non-trivial solution~$\tilde \D$ depending on~$\D$, and  \eqref{mainRelation} indeed defines a multivalued map $\D \mapsto \tilde \D$. To identify the latter with the pentagram map, we need to rewrite it in terms of bi-infinite sequences $V$, $\tilde V$ annihilated by the operators $ \D$ and $\tilde \D$ respectively. Applying both sides of \eqref{mainRelation} to $V$, we get
$$
\tilde{\D}_+  {\D}_-V = \tilde {\D}_-  {\D}_+V,
$$
which, using that $\D V= 0$ and thus $ \D_-V= -\D_+V$, can be rewritten as
$$
\tilde {\D}\,  {\D}_+ V = 0.
$$
But the latter means that 
$
\tilde V =  {\D}_+ V,
$
which is exactly the definition of the pentagram map. Indeed, by definition of $\D_+$, the vector
$({\D}_+ V)_i$ belongs to the span $\langle V_{i}, V_{i+2}\rangle$ of $V_{i}, V_{i+2}$.
At the same time, we have ${\D}_+ V = -{\D}_- V$, so
$$
({\D}_+ V)_i = -({\D}_- V)_i   \in\langle V_{i+1}, V_{i+3}\rangle.
$$
Therefore, we have
$$
({\D}_+ V)_i \in \langle V_{i}, V_{i+2}\rangle \cap \langle V_{i+1}, V_{i+3}\rangle,
$$
which means that the corresponding point in $\RP^2$ is the intersection of consecutive shortest diagonals~$\langle v_i, v_{i+2}\rangle$ and $\langle v_{i+1}, v_{i+3}\rangle$, as desired.
\end{proof}
\begin{corollary}
The pentagram map, written in terms of difference operators, is a refactorization relation.
\end{corollary}
\begin{proof}
Relation { \eqref{mainRelation}} can be rewritten as
\begin{equation}\label{refactRelation}
\tilde{\D}_-^{-1}\tilde{\D}_+ = \D_+  \D_-^{-1},
\end{equation}
where the inverses of difference operators are understood as \textit{pseudo-difference operators}. To see that this formula defines a refactorization mapping, consider the operator $\Lax:= {\D}_-^{-1}{\D}_+$.
Then~\eqref{refactRelation} means that the dynamics of $\Lax$ under the pentagram map is given by $\Lax \mapsto \tilde{ \Lax} $, where 
$
\tilde{ \Lax} := \D_+  \D_-^{-1}.
$
 Therefore, the pentagram map in terms of $\Lax$ is a refactorization map
\begin{equation*}\tag*{\qedhere} 
{\D}_-^{-1}{\D}_+ \mapsto \D_+  \D_-^{-1}.
\end{equation*}
\end{proof}
A crucial part of the proof of Theorem \ref{thm0} is solvability of \eqref{mainRelation} with respect to $\tilde \D$, which in turn is related to a very special choice of exponents of the shift operator $\Sh$ entering  $ \D_-$ and $ \D_+$. We refer to the set of integers that are the exponents of $\Sh$ entering a given difference operator $\D$ as the \textit{support} of~$\D$. It is easy to see that~\eqref{mainRelation} is solvable if and only if the supports $J_\pm \subset \Z$ of the operators $ \D_\pm$ satisfy
$$
|J_+ + J_-| < |J_-| + |J_+|,
$$
where $J_+ + J_-$ is the Minkowski sum. Furthermore, for sets $J_\pm$ with $|J_\pm| > 1$ the latter inequality holds if and only if the $J_\pm$  are finite {arithmetic progressions with the same common difference}. Different choices of such pairs of progressions lead to different pentagram-type maps admitting a refactorization description. As we already saw, the choice $\{0,2\}$, $\{1,3\}$ corresponds to the usual pentagram map. More generally, the choice $\{0,2, 4, \dots\}$, $\{1,3, 5, \dots\}$ corresponds to short-diagonal maps of~\cite{khesin2013}. Similarly, $\{0,1\}$, $\{2,3\}$ leads to the inverse pentagram map, while $\{0,1, \dots, p\}$, $\{p+1, p+2, \dots, q\}$ corresponds to the inverse dented map of \cite{khesin2016}. Finally, the choice $\{0,d\}$, $\{1,d+1\}$ leads to the pentagram map on corrugated polygons in $\RP^d$ studied in~\cite{Gekhtman2016}.\par
One can also consider relation \eqref{mainRelation} for difference operators $\D_\pm$ with non-disjoint supports. Such maps still admit a refactorization description, but they do not have a pentagram-like interpretation. Indeed, in this case the pair $\D_\pm$ is not equivalent to a single operator $ \D_+ +  \D_- $, and because of that the phase space cannot be interpreted as the space of polygons. The simplest case $\{0,1\}, \{1,2\}$ corresponds to the leapfrog map defined in  \cite{Gekhtman2016}, while for other cases of non-disjoint supports the geometric interpretation is not known.


\par

The structure of the paper is as follows. In Section \ref{sec:geometry} we define a general class of pentagram-type maps associated with pairs of disjoint arithmetic progressions with the same common difference. This class, in particular, includes all previously known integrable pentagram-type maps. In Section \ref{sec:diff} we discuss difference and pseudo-difference operators, along with Poisson structures on such operators. 
Section \ref{sec3} contains main results of the paper, namely we show that pentagram maps of Section \ref{sec:geometry} fit into an even bigger class of dynamical systems which are parametrized by pairs of not necessarily disjoint progressions and admit a refactorization description. As a corollary, all such maps admit an invariant Poisson structure and a Lax representation with Poisson-commuting spectral invariants. It is therefore very likely that all these maps are both Liouville and algebraically integrable. This integrability problem will be addressed in a separate publication. In addition to these results, in Section~\ref{sec3}  we also discuss some applications, as well as relations to known constructions. In particular, in Section~\ref{sec:scaling} we show how our approach yields the \textit{scaling invariance} of pentagram-type maps, which was the central tool in the proof of integrability for the classical, as well as for short-diagonal and dented maps. Further, in Section~\ref{sec:poisson} we explicitly compute Poisson brackets for the short-diagonal pentagram map in $\RP^3$. In contrast to previously known cases, those brackets turn out to be not quadratic but polynomial of degree four. After that, in Section \ref{sec:meshes} we outline the connection between the approach of the present paper and the Y-meshes description of higher pentagram maps given in \cite{glick2015}. It turns out that Y-meshes are related to factorizations of difference operators. In addition to that, in Section \ref{sec:networks} we show how our refactorization approach can be used to represent pentagram-type maps using moves in Postnikov networks, as in \cite{Gekhtman2016}. This also gives a cluster description of the pentagram map, and one may hope to use our approach to extend the cluster algebra formalism to multidimensional maps, which is still an open problem. Finally, Section \ref{sec:problems} is devoted to open questions.

\par
\bigskip

{\bf Acknowledgments.} The author is grateful to Quinton Aboud, Anton Alekseev, Michael Gekhtman, Boris Khesin, Nicolai Reshetikhin, Richard Schwartz, Alexander Shapiro, Yuri Suris, Sergei Tabachnikov, and the anonymous referee for useful discussions and remarks. A large part of this work was done during the author's visit to Max Planck Institute for Mathematics (Bonn). The author would like to thank the Institute's faculty and staff for their support and stimulating atmosphere.  This work was supported by NSF grant DMS-2008021.

\par\medskip
\section{Pentagram-type maps associated with pairs of arithmetic progressions}\label{sec:geometry}
In this section we explain how to associate a pentagram-type map to any pair of finite disjoint arithmetic progressions $J_\pm \subset \Z$ with the same common difference. As particular cases of this construction, one obtains all known integrable pentagram-type maps. 
Later on, in Section \ref{sec:mainThm}, we will show that these maps fit into a more general class of dynamical systems  which are parametrized by pairs of not necessarily disjoint progressions and admit a refactorization description.\par
All pentagram-type maps operate on polygons, i.e. ordered sequences of points in the projective space. We will only consider polygons satisfying the following natural condition:
\begin{definition}\label{def:poly}
A \textit{polygon} in $ \RP^d$ is a bi-infinite sequence of points $\{ v_i \in \RP^d\}$ such that any $d+1$ consecutive points $v_i, \dots, v_{i+d+1}$ are in general position (i.e. do not belong to a subspace of dimension $d-1$).
\end{definition}
In contrast to the classical pentagram map, which is well-defined for all generic polygons, some of the maps that we will study operate on a more restricted class of polygons whose vertices satisfy certain additional coplanarity conditions, described in the following definition:
\begin{definition}\label{jcorr}
Let $J \subset \Z$, $|J| \geq 2$ be a finite set of integers containing at least two elements, and let $d := \max(J) - \min(J) - 1$. Then a polygon $\{ v_i\}$ in $\RP^d$ is called \textit{$J$-corrugated}  if for any $i \in \Z$ the points $\{ v_{i+j} \mid j \in J\}$ belong to a $|J|-2$ dimensional plane (instead of a $|J|-1$ dimensional plane, which is the generic case).
\end{definition}
\begin{example}\label{Jisint}
Assume that $J$ consists of consecutive integers, $J = \{j, j+1, \dots, j+d+1\}.$ Then  a {$J$-corrugated polygon} is any polygon in $\RP^d$.
\end{example}
\begin{example}
Assume that $J = \{ 0,1,d,d+1\}$. Then $J$-corrugated polygons are \textit{corrugated polygons} in $\RP^d$ in the sense of \cite[Section 5.1.1]{Gekhtman2016}. 
\end{example}
\begin{example}
Assume that $J = \{ 0,1, \dots, l\} \cup \{m, m+1, \dots, d+1\}$, where $l < m$, is a union of two disjoint sets of consecutive integers. Then $J$-corrugated polygons are \textit{partially corrugated polygons} in $\RP^d$ in the sense of~\cite[Definition 6.3]{khesin2016}. 
\end{example}
We now define an analogue of the pentagram map on the space  of $J$-corrugated polygons. Such a map can be defined if $J \subset \Z$ can be partitioned as $J = J_+ \sqcup J_-$, where $J_\pm \subset \Z$ are finite {arithmetic progressions with the same common difference}.


\begin{definition}
Let $J_\pm \subset \Z$ be non-empty disjoint finite integral arithmetic progressions with the same common difference. Let also $J := J_+ \cup J_-$. Then the \textit{pentagram map associated with the pair $J_\pm$} is the map from the space of $J$-corrugated $n$-gons to itself defined by
\begin{equation}\label{jmap}
\tilde v_i : = \langle v_{i + j} \mid j \in J_+ \rangle \cap \langle v_{i + j} \mid j \in J_-\rangle.
\end{equation}
Here $v_i$'s are the vertices of the initial polygon, $\tilde v_i$'s are vertices of its image under the map, and the notation $\langle v_i \rangle$ stands for the projective subspace spanned by the points $\{v_i\}$.
\end{definition} 
\begin{remark}
This definition makes sense for arbitrary disjoint finite sets $J_\pm \subset \Z$, but for general  $J_\pm$ the image of a $J$-corrugated polygon (where $J := J_+ 
\cup J_-$) under the so defined map is not  $J$-corrugated. This property, however, does hold if $J_\pm $ are arithmetic progressions with the same common difference, as shown by the following proposition.

\end{remark}
\begin{proposition}
For any non-empty disjoint finite arithmetic progressions $J_\pm \subset \Z$ with the same common difference, the corresponding pentagram map is a generically well-defined mapping from the space of $J$-corrugated $n$-gons to itself.
\end{proposition}
\begin{proof} For a generic $J$-corrugated polygon $\{v_i \in \RP^d\}$, where $d = \max(J) - \min(J) - 1$,  the subspaces $\langle v_{i + j} \mid j \in J_\pm \rangle$ have complementary dimensions $ |J_\pm| - 1$ in the space $\langle v_{i + j} \mid j \in J \rangle$ of dimension $|J| - 2$. 
Therefore, their intersection indeed defines a point $ \tilde v_i\ \in \RP^d$. Furthermore, it is not hard to see that for generic $v_i$'s any $d+1$ consecutive points $\tilde v_i$ will be in general position, so $\{ \tilde v_i\}$ is a polygon in the sense of Definition \ref{def:poly}. 
Thus, it remains to show that the new polygon $\{ \tilde v_i\}$ is $J$-corrugated. To that end, for any $i \in \Z$, consider the subspace
$
L_i :=  \langle v_{i + j} \mid j \in J_+ + J_-\rangle,
$
where $J_+ + J_-$ is the Minkowski sum of $J_+$ and $J_-$.  Notice that for any $i \in \Z$ and any $j \in J$, we have $\tilde v_{i+j} \in L_i$. Indeed, by construction of the polygon  $\{ \tilde v_i\}$, we have
$$
\tilde v_{i+j} = \langle v_{i + j + j'} \mid j' \in J_+ \rangle \cap \langle v_{i + j + j'} \mid j' \in J_-\rangle.
$$
Assume that $j \in J_+$. Then
$\langle v_{i + j + j'} \mid j' \in J_-\rangle$ is a subspace of $L_i$, because $j + j'  \in J_+ + J_-$. Therefore, $\tilde v_{i+j} \in L_i$. Analogously, if $j \in J_-$, then $\langle v_{i + j + j'} \mid j' \in J_+ \rangle$  is a subspace of $L_i$, and we still have $\tilde v_{i+j} \in L_i$. So, all the points $\{ \tilde v_{i+j} \mid j \in J \}$ belong to $L_i$. But the dimension of $L_i$ does not exceed
$$
|J_+ + J_-| - 1 = |J_+| + |J_-| - 2 = |J| -2,
$$
where we use that for finite arithmetic progressions  $J_\pm \subset \Z$ with the same common difference one has $|J_+ + J_-| = |J_+| + |J_-| - 1$. So, for every $i \in \Z$, the points $\{ \tilde v_{i+j} \mid j \in J \}$ belong to at most $|J| -2$ dimensional subspace $L_i$, which means the polygon $\{ \tilde v_i\}$ is indeed $J$-corrugated, as desired.
\end{proof}
\renewcommand{\arraystretch}{1.2}
\begin{table}
$$
\left.\begin{array}{c|c|c}   && \\[-1em] \mathbf{J_+} & \mathbf{J_-} & \mbox{\textbf{The corresponding map}}\\[-1em] && \\ \hline \{0,2\} & \{1,3\} & \mbox{Classical pentagram map} \\\hline \{0,1\} & \{2,3\} & \mbox{Inverse pentagram map} \\\hline \{0,d\} & \{1,d+1\} & \mbox{Pentagram map on corrugated polygons in $\RP^d$ \cite{Gekhtman2016}} \\\hline \{0,1\} & \{d, d+1\} & \mbox{Inverse pentagram map on corrugated polygons in $\RP^d$} \\\hline \{0, \dots, k\} & \{k+1, \dots, d+1\} & \mbox{Inverse dented pentagram maps in $\RP^{d}$ \cite{khesin2016}}  \\\hline 
\{0,2,4, \dots, 2k\} & \{1,3,5, \dots, 2k+1\} & \mbox{Short-diagonal pentagram map in $\RP^{2k}$ \cite{khesin2013}   } 
\\\hline 
\{0,2,4, \dots, 2k\} & \{1,3,5, \dots, 2k-1\} & \mbox{Short-diagonal pentagram map in $\RP^{2k-1}$  \cite{khesin2013}   } \\

 \end{array}\right.
$$
\caption{Examples of pentagram maps associated with pairs of arithmetic progressions.}\label{tab:pentmaps}
\end{table}
\begin{example}
Examples of pentagram maps associated with pairs of arithmetic progressions are given in Table~\ref{tab:pentmaps}. Note that these examples cover all known integrable cases, so all such cases fit into the above construction. 
\end{example}
\begin{remark}\label{rem:shifting}
The classical pentagram is usually defined by $\tilde v_i = \langle v_{i-1}, v_{i+1} \rangle \cap \langle v_{i}, v_{i+2} \rangle$ (\textit{right} labelling scheme), or by $\tilde v_i = \langle v_{i-2}, v_{i} \rangle \cap \langle v_{i-1}, v_{i+1} \rangle $ (\textit{left} labelling scheme). This corresponds to progressions $\{-1,1\}, \{0,2\}$ for the right scheme, and $\{-2,0\}, \{-1,1\}$ for the left scheme. Our choice $\{0,2\}, \{1,3\}$ corresponds to the same map, but with a different labeling of vertices of the resulting polygon. More generally, shifting both $J_+$ and $J_-$ by the same number results in the same map up to a shift of indices. 
\end{remark}
\begin{remark}
Note that except for the short-diagonal and inverse dented cases, our construction gives no maps which are defined on all generic polygons (with no additional coplanarity conditions). Indeed, such maps would correspond to sets $J$ consisting of consecutive integers (cf. Example \ref{Jisint}), and without loss of generality we can assume that $J= \{0,1,\dots \}$ (because we can always shift $J$, as in Remark~\ref{rem:shifting}). But the only ways to represent this set $J$ as a disjoint union of two arithmetic progressions with the same common difference are $\{0, 1, \dots, k\} \sqcup \{k+1, k+2, \dots\}$ and $\{0, 2, 4, \dots\} \sqcup \{1,3,5, \dots\}$, which corresponds to the inverse dented and short-diagonal maps respectively. 
\end{remark}
The space of $J$-corrugated polygons is infinite-dimensional for any $J$ with $|J| > 2$. One can still study pentagram-type maps on such spaces, but to obtain integrable dynamics one should impose some kind of boundary conditions on the polygon $\{v_i\}$. From the geometric perspective, the most natural condition is closedness, $v_{i+n} = v_i$. However, it turns out that the pentagram map, as well as similar maps studied in the present paper, have much better properties on a bigger space of polygons that are closed only up to a projective transformation. Such polygons as known as {twisted}:
 \begin{definition}
A \textit{twisted $n$-gon} is a polygon $\{ v_i \in \RP^d\}$ such that $v_{i+n} = \phi(v_i)$ for every $i$ and a fixed (not depending on $i$) projective transformation $\phi \colon \RP^d \to \RP^d$, called the \textit{monodromy}.
\end{definition}
It is clear that all pentagram maps defined above (as well as any other map on polygons which is defined using only projectively natural operations) take twisted polygons to twisted polygons and, moreover, preserve the monodromy. Throughout the paper, all pentagram-type maps are assumed to operate on twisted polygons.

\medskip

\section{Difference and pseudo-difference operators}\label{sec:diff}

\subsection{Generalities on difference operators}\label{sec:primer}
In this section we recall some basic notions related to difference operators. Our terminology mainly follows that of \cite{van1979spectrum}. Let $\R^\infty$ be the vector space of bi-infinite sequences of real numbers, and let $J \subset \Z$ be a finite collection of integers. A linear operator $\D \colon \R^\infty \to \R^\infty$ is called \textit{a difference operator supported in $J$} if it can be written as
\begin{equation}\label{dodef}
(\D\xi)_i = \! \sum_{j \in J} \! a_{j,i} \xi_{i+j},
\end{equation}
or, equivalently, if
	\begin{align}\label{genDiffOp}
\D = \!\sum_{j \in J} a_j T^j,
\end{align}
where $T \colon \R^\infty \to \R^\infty$ is the left shift operator $(T\xi)_i = \xi_{i+1}$, and each coefficient $a_j$ is a bi-infinite sequence $\{ a_{j,i} \mid i \in \Z\}$ of real numbers acting on $\R^\infty$ by term-wise multiplication. Such sequences can per se be regarded as difference operators with $J = \{0\}$. 
\par
The \textit{order} of difference operator \eqref{dodef} is the number $\ord \D := M - m$, where $M: = \max J$, $m:=\min J$. Difference operator~\eqref{dodef} is called \textit{properly bounded} if  none of the elements of sequences $a_{m}$, $a_{M}$ vanish. Clearly, for a properly bounded difference operator $\D$ one has
$
\dim \Ker \D = \ord \D.
$ 
A difference operator $\D$ is \textit{$n$-periodic} if all its coefficients $a_j$ are $n$-periodic sequences, which is equivalent to saying that $\D$ commutes with the $n$'th power of the shift operator: $\D T^n = T^n\D$. Clearly, if $\D$ is an $n$-periodic operator, then its kernel is invariant under the action of $T^n$. The finite-dimensional operator $T^n\vert_{\Ker \D}$ is called the \textit{monodromy} of $\D$. Eigenvectors of the monodromy operator $ T^n\vert_{\Ker \D}$ are exactly {quasi-periodic} solutions of the equation $\D\xi = 0$, i.e. solutions which belong to the space of quasi-periodic sequences
\begin{equation}\label{qspace}
 \{ \xi \in \R^\infty \mid \xi_{i+n} = z\xi_i\}
\end{equation}
for certain $z \in \R^*$. 
\par
We denote the space of $n$-periodic difference operators supported in $J$ by $\DO{n}{J}$, while $\PBDO{n}{J} \subset \DO{n}{J}$ stands for the (dense) subset of properly bounded operators. Let also $\ALLDO{n}$ be the associative algebra of all $n$-periodic difference operators (with arbitrary finite support).
\begin{remark}\label{app:loopgroup}
The algebra $\ALLDO{n}$ of  $n$-periodic difference operators is isomorphic to the algebra $\gl_n \otimes \R[z,z^{-1}]$ of $\gl_n$-valued Laurent polynomials in one variable~$z$ (here $\gl_n$ stands for the associative algebra of $n \times n$ matrices over the base field $\R$). Indeed, consider the natural action of $n$-periodic difference operators on the space \eqref{qspace} of all $n$-quasi-periodic bi-infinite sequences of real numbers with monodromy $z$. This gives a $1$-parametric family $\rho_z$ of $n$-dimensional representations of the algebra $\ALLDO{n}$.
In each of the spaces \eqref{qspace}, take a basis  $\xi_1, \dots, \xi_n$ determined by the condition $\xi_{ij}= \delta_{ij}$ for $i,j = 1, \dots, n$ (where $\delta_{ij}$ is the Kronecker delta). Written in this basis, the representation $\rho_z$ takes an $n$-periodic sequence $a = \{a_i\}$ (viewed as a zero order difference operator) to a diagonal matrix with entries $a_1, \dots, a_n$, while the shift operator $T$ becomes the matrix $\sum_{i=1}^{n-1} {E}_{i,i+1} + z{E}_{n,1}$, where ${E}_{i,j}$ is the matrix with a $1$ at position $(i,j)$ and zeros elsewhere.
Therefore, since the algebra of difference operators is generated by sequences, $T$, and $T^{-1}$, it follows that $\rho_z$ can be viewed as a homomorphism of difference operators into
 $ {\gl}_n \otimes \R[z, z^{-1}]$. Furthermore, it is easy to verify that this homomorphism is a bijection, and hence an isomorphism. 
 
\end{remark}
\begin{proposition}\label{prop:detmono}
Let $\D$ be a properly bounded difference operator supported in $J$, and let $\D(z)$ be the associated element of the loop algebra. Then $\det \D(z)$ is a constant multiple of the polynomial $z^{m}\CPM{\D}{z}$, where $m:= \min J$, and $\CPM{\D}{z}$ is the characteristic polynomial of the monodromy of $\D$.
\end{proposition}
\begin{proof}
If we multiply $\D$ by $T^k$, where $k \in \Z$, then the characteristic polynomial of its monodromy does not change, while the polynomial $\det \D(z)$ gets multiplied by $\det T^k(z) = (\det T(z))^k = z^k$. So, it suffices to consider the case $\min J = 0$. Furthermore, it is sufficient to prove the statement for generic properly bounded operators supported in $\{0, \dots, d\}$, because within that set the coefficients of both polynomials $\det \D(z)$ and $\CPM{\D}{z}$ are polynomial functions in terms of the coefficients of $\D$. So, if one can show that these polynomials are proportional for generic operators, then it must be true for all operators. To establish the statement for generic $\D$, observe that by definition of $\D(z)$ the polynomial   $\det \D(z)$ vanishes for some $z \neq 0$ if and only if $\D$ has a kernel on the space \eqref{qspace}, which is equivalent to saying that $z$ is an eigenvalue of the monodromy of $\D$. So, the roots of the polynomials $\det \D(z)$ and $\CPM{\D}{z}$ are the same (as sets). Furthermore, for generic $\D$ all roots of $\CPM{\D}{z}$ are distinct. So, to prove that the polynomials $\det \D(z)$ and $\CPM{\D}{z}$  are proportional, it suffices to show that they have the same degree. In other words, we need to show that the degree of $\det \D(z)$ is equal to the degree $d$ of $\D$. This can be checked by explicitly writing down the matrix $\D(z)$, or by using the following argument. First of all, one easily checks that the statement holds for operators of degree~$1$. But a generic operator $\D$ of degree $d$ can be written as a product of operators of degree~$1$, so by multiplicativity for such operator we have $\deg \det \D(z) = d$, as desired.
\end{proof}

\subsection{Difference operators and $J$-corrugated polygons}\label{subsec:diffpoly}
There is a close relation between difference operators supported in $J$ and {$J$-corrugated} polygons.
%
 Denote by $\P_n(J)$ the space of twisted $J$-corrugated $n$-gons, and let $\P_n(J) \,/\, \PGL$ be the quotient of that space by projective transformations. 
We will describe that space as a certain quotient of the space $\PBDO{n}{J}$ of properly bounded $n$-periodic difference operators supported in $J$. Namely, let $H$ be the group of non-vanishing $n$-quasi-periodic scalar sequences, i.e.
$$
H := \{ \alpha \in \R^\infty \mid \forall\, i\in \Z,  \alpha_i \neq 0, \mbox{ and } \exists \,  z \in \R^* \mbox{ s.t. }  \forall\, i\in \Z, \,\alpha_{i+n} = z\alpha_i \}. 
$$
Further, let $H \ttimes H$ be the subgroup of $H \times H$ that consists of pairs of  non-vanishing $n$-quasi-periodic scalar sequences with the same monodromy, i.e.
\begin{align*}
H \ttimes H := &\{ (\alpha, \beta) \in \R^\infty \times \R^\infty \mid \forall\, i\in \Z,  \alpha_i \neq 0, \beta_i \neq 0, \\ &\mbox{ and } \exists \,  z \in \R^* \mbox{ s.t. }  \forall\, i\in \Z, \,\alpha_{i+n} = z\alpha_i, \beta_{i+n} = z\beta_i \}. 
\end{align*}
This group acts on the space $\DO{n}{J}$ of $n$-periodic difference operators with given support by means of the \textit{left-right action}
\begin{equation}\label{action}
\D \mapsto \alpha \D \beta^{-1}.
\end{equation}
\begin{proposition}\label{prop:quo}
For any finite subset $J \subset \Z$ with $|J| \geq 2$, there is a one-to-one correspondence (a homeomorphism) between the following spaces:
\begin{enumerate}
\item The space $\P_n(J) \,/\, \PGL$ of twisted $J$-corrugated $n$-gons modulo projective transformations.
\item The space $\PBDO{n}{J} \, / \,H \ttimes H$ of properly bounded $n$-periodic difference operators supported in $J$ modulo the left-right action \eqref{action} of the group $H \ttimes H$ of pairs of  non-vanishing $n$-quasi-periodic scalar sequences with the same monodromy.
\end{enumerate}

\end{proposition}
\begin{proof}
The proof is analogous to that of \cite[Proposition 2.2]{izosimov2019pentagram}. Let us just briefly outline the construction. Given a projective equivalence class of generic  twisted $J$-corrugated $n$-gons, consider an arbitrary representative $\{v_i \in \RP^d\} \in \P_n(J)$ of that class (here $d := \max(J) - \min(J) - 1$). Lift the quasi-periodic sequence of points $v_i \in \RP^d$ to a quasi-periodic sequence of vectors $V_i\ \in \R^{d+1}$. Then, from the $J$-corrugated condition it follows that for any $i \in \Z$ the vectors $\{V_{i+j}\mid j \in J\}$ belong to a subspace of dimension $|J|-1$ and, therefore, are linearly dependent:
\begin{equation}\label{videp}
\sum_{j \in J} \! a_{j,i} V_{i+j} = 0.
\end{equation}
This is equivalent to $\D V = 0$, where $V$ is the bi-infinite sequence of $V_i$'s, and the operator $\D$ is given by~\eqref{dodef}. Furthermore, since the sequence $\{V_i\}$ is quasi-periodic, the so-obtained operator $\D$ is periodic, while from the genericity condition for $\{v_i \in \RP^d\}$ it follows that $\D$ is properly bounded. Hence we obtain a properly bounded $n$-periodic difference operator supported in $J$. To complete the proof, it suffices to notice that from the possibility to rescale each of the $V_i$'s and also multiply each of the equations~\eqref{videp} by a scalar it follows that $\D$ is defined up to the left-right action  \eqref{action}. Details of the proof (in the case when $J$ consists of four consecutive integers) can be found in \cite{izosimov2019pentagram}.
\end{proof}
As can be seen from this construction, one has the following relation between the monodromy of a $J$-corrugated polygon and the monodromy of the corresponding difference operator:
\begin{corollary}
Let  $P \in \P_n(J)$ be a twisted $J$-corrugated $n$-gon, and let $\D \in \PBDO{n}{J}$ be one of the corresponding difference operators. Then the monodromy of $P$ is conjugate to the projectivization of the monodromy of $\D$.

\end{corollary}
\begin{proof}
Assume that the monodromy of the polygon $\{v_i\}$  in the proof of Proposition \ref{prop:quo} is given by the projective transformation $\phi$. Then the sequence of $V_i$'s satisfies $V_{i+n} = MV_i$, where $M$ is a matrix of $\phi$ (i.e. $\phi$ is the projectivization of $M$). At the same time, the components of the vectors $V_i$ form a basis in the space $\Ker \D$, and the monodromy matrix $T^n\vert_{\Ker \D}$ written in that basis is the transpose of $M$ (indeed, denoting the basis vectors by $\xi_1, \dots, \xi_d$, we can rewrite $V_{i+n} = MV_i$ as $ (T^n \xi_1, \dots, T^n \xi_d)^t  = M(\xi_1, \dots, \xi_d)^t $, which means that the matrix of the transformation $T^n$ in the basis $\xi_1, \dots, \xi_d$ is $M^t$). So, the monodromy of the polygon $\{v_i\}$  is the projectivization of $(T^n\vert_{\Ker \D})^t$, and hence is conjugate to the projectivization of  $T^n\vert_{\Ker \D}$.
\end{proof}
In particular, one has the following relation between the eigenvalues of the monodromy and the determinant of the corresponding loop algebra element:
\begin{corollary}\label{monodet}
Let  $P \in \P_n(J)$ be a twisted $J$-corrugated $n$-gon, and let $\D \in \PBDO{n}{J}$ be one of the corresponding difference operators. Then the eigenvalues of the monodromy of $P$ coincide (with multiplicities) with non-zero roots of the polynomial $\det \D(z)$, where $\D(z)$ is an element of the loop algebra corresponding to the difference operator $\D$, as described in Remark \ref{app:loopgroup}.
\end{corollary}
\begin{proof}
This follows from Proposition \ref{prop:detmono}.
\end{proof}
\begin{remark}
Note that the monodromy of a twisted polygon is a projective transformation, so its eigenvalues are defined up to simultaneous multiplication by the same constant. However, the same is true for the roots of  $\det \D(z)$, because taking $\alpha$ and $\beta$ in \eqref{action} with non-trivial monodromy $w$ leads to simultaneous rescaling of all the roots by a factor of $w$.

%
\end{remark}

\subsection{The Poisson-Lie group of pseudo-difference operators}
We define an $n$-periodic \textit{pseudo-difference operator} as a formal Laurent series in terms of the left shift operator $T$, whose coefficients are $n$-periodic sequences. In other words, every such operator is of the form
\begin{equation}\label{app:psido}
\sum_{j = k}^{+\infty} a_j T^j,
\end{equation}
where $k \in \Z$ is an integer, $T$ is the left shift operator on $\R^\infty$, while each $a_j$ is an $n$-periodic bi-infinite sequence of real numbers. Such an expression can be regarded either as a formal sum, or as an actual operator acting on the space $\{\xi \in \R^\infty \mid \exists\, j \in \Z : \xi_i = 0 \, \forall \, i > j\}$ of eventually vanishing sequences.\par
We will denote the set of $n$-periodic pseudo-difference operators by $\PSIDO{n}$. Is is an associative algebra with respect to addition and multiplication (composition) of operators. Moreover, almost every pseudo-difference operator in invertible. In particular, \eqref{app:psido} is invertible if the coefficient $a_k$ of lowest power in $T$ is a sequence none of whose elements vanish. 
We will denote the set of invertible $n$-periodic pseudo-difference operators by $\IPSIDO{n}$. This is a group with respect to multiplication. It can be regarded as an infinite-dimensional Lie group with Lie algebra $\PSIDO{n}$.

\begin{remark}
One can also consider (and apply for the purposes of the present paper) pseudo-difference operators which have infinitely many terms of negative degree in $T$, but only finitely many terms of positive degree. This leads to an isomorphic algebra. 
\end{remark}
\begin{remark}\label{rm:psiloop}
The isomorphism $\ALLDO{n}\simeq\gl_n \otimes \R[z,z^{-1}]$ described in Remark \ref{app:loopgroup} naturally extends to an isomorphism between the algebra of $n$-periodic pseudo-difference operators, and the algebra $\gl_n \otimes \R((z))$ of matrices over the field $\R((z))$ of formal Laurent series with real coefficients and finitely many terms of negative degree. Under this isomorphism, the group $\IPSIDO{n}$ of invertible pseudo-difference operators is identified with the group of matrices over Laurent series with non-vanishing determinant (this group is one of the versions of the \textit{loop group} of $\GL_n$). 
\end{remark}
\begin{proposition}\label{app:mainprop}
There exists a natural Poisson structure $\pi$ on the group $\IPSIDO{n}$ of $n$-periodic invertible pseudo-difference operators. This structure has the following properties:
\begin{enumerate}
\item It is multiplicative, in the sense that the group multiplication is a Poisson map. In other words, the group $\IPSIDO{n}$, together with the structure $\pi$, is a \textit{Poisson-Lie group}.
\item Assume that $J \subset \Z$ is a finite subset that consists of consecutive integers. Then the subset $\IDO{n}{J} := \IPSIDO{n} \cap \DO{n}{J}$ of invertible difference operators (that is, difference operators whose inverse is well-defined as a \textit{pseudo-difference} operator) supported in $J$ is a Poisson submanifold of  $\IPSIDO{n}$.
\item If $J$ is a one-point set, then the restriction of $\pi$ to $\IDO{n}{J}$ is zero. In particular, the Poisson structure $\pi$ vanishes on sequences (viewed as difference operators supported in $\{0\}$).
\item The Poisson structure $\pi$ is invariant under the left-right action \eqref{action} of the group $H \ttimes H$ of pairs of  non-vanishing $n$-quasi-periodic sequences with the same monodromy. 
%
\item Central functions on $\IPSIDO{n}$ Poisson commute. 
\end{enumerate}
\end{proposition}
\begin{remark}
As explained in Remark \ref{rm:psiloop}, the group  $\IPSIDO{n}$ is isomorphic to a version of the loop group of $\GL_n$. In the loop group language, the Poisson structure $\pi$ is well-known: it is the one associated with the \textit{trigonometric $r$-matrix}. Here we will provide a construction of this Poisson structure which does not appeal to the loop group formalism. In fact, the language of (pseudo)difference operators seems to be more natural when dealing with the trigonometric $r$-matrix. We will, however, use the loop group language in some of the computations, see in particular Section \ref{sec:gln}.\par
Our  Poisson structure on pseudo-difference operators can also be viewed as a natural discrete analogue of the Poisson-Lie structure on pseudo-differential operators \cite{khesin1995poisson}.
\end{remark}

\begin{remark}
For periodic sequences $\alpha$ and $\beta$, the fourth statement of Proposition \ref{app:mainprop} follows from the third one combined with the first. Indeed, by the third statement the Poisson structure $\pi$ vanishes on sequences, so from multiplicativity we get that both left and right multiplications by sequences are Poisson maps. However, if the sequences $\alpha$ and $\beta$ have non-trivial monodromy, then one cannot extract the fourth statement of the proposition from multiplicativity, because in that case $\alpha$ and $\beta$ are not elements of $\IPSIDO{n}$.

%

\end{remark}
\begin{remark}\label{ratvsint}
Take a subset $J \subset \Z$ which consists of $d+2$ consecutive integers. Then the $J$-corrugated condition is vacuous and, according to Proposition \ref{prop:quo}, the quotient of  $\PBDO{n}{J}$ by the action \eqref{action} can be identified with the space of twisted polygons in $\RP^d$, considered up to projective equivalence. So, restricting the Poisson structure $\pi$ to $\PBDO{n}{J}$ (which is an open subset of $\IDO{n}{J}$ and hence a Poisson submanifold) and taking the quotient under the action \eqref{action} one gets a Poisson structure on the space of polygons. It seems, however, that this structure has nothing to do with pentagram maps. As will be explained below, Poisson structures invariant under pentagram maps arise from Poisson submanifolds of $\IPSIDO{n}$ given by \textit{rational} pseudo-difference operators, i.e. operators that can be written as a quotient of two difference operators. 
\end{remark}
We prove Proposition \ref{app:mainprop} in  Section \ref{app:mainproof}, after a brief general discussion of Poisson-Lie groups in Section \ref{sec:gplg}. 

\subsection{Generalities on Poisson-Lie groups}\label{sec:gplg}
This section is a brief introduction to the theory of Poisson-Lie groups. Our terminology follows that of~\cite{reiman2003integrable}. Recall that a Lie group $G$ endowed with a Poisson structure $\pi$ is called a \textit{Poisson-Lie group} if $\pi$ is \textit{multiplicative}, i.e. if the multiplication $G \times G \to G$ is a Poisson map (it also follows from this that the inversion map $i \colon G \to G$ is anti-Poisson, i.e. $i_*\pi = -\pi$). Assume that $G$ is a Poisson-Lie group, and let $\g$ be its Lie algebra. Then, by considering the left trivialization of the tangent bundle of $G$, one can identify the bivector field $\pi$ with a map $G \to \g \wedge \g$. Furthermore, one can show that multiplicativity of $\pi$ is equivalent to that map being a cocycle on $G$ with respect to the adjoint representation of $G$ on $ \g \wedge \g$. If that cocycle is a coboundary, then $G$ is called a \textit{coboundary Poisson-Lie group}. A Poisson-Lie group $G$ is coboundary  if and only if there exists an element $\hat r \in \g \wedge \g$, called the \textit{classical $r$-matrix}, such that the Poisson tensor $\pi$ at every point $g \in G$ is given by
\begin{equation}\label{app:cobound}
\pi_g = \frac{1}{2} \left(\vphantom{\int} (\lambda_g)_*\hat r - (\rho_g)_*\hat r\right),
\end{equation}
where $\lambda_g$ and $\rho_g$ are, respectively, the left and right translations by $g$. 
Note that although the bivector~\eqref{app:cobound} is automatically multiplicative (since any coboundary is a cocycle), it does not need to satisfy the Jacobi identity. The necessary  and sufficient condition for \eqref{app:cobound} to satisfy the Jacobi identity is a rather complicated equation in terms of $\hat r$ which is usually replaced by simpler sufficient conditions, such as the \textit{modified Yang-Baxter equation}. 
We will state this condition under the assumption that the Lie algebra $\g$ is endowed with an invariant (under the adjoint action of $G$) inner product, in which case one can identify the bivector $\hat r \in \g \wedge \g$ with a skew-symmetric operator $r \colon \g \to \g$ (which is also called the $r$-matrix). In terms of that operator, the modified Yang-Baxter equation reads
\begin{equation}\label{app:mybe}
[rx, r y] - r[rx, y] - r[x, r y] = -[x,  y] \quad \forall\, x,y \in \g.
\end{equation}
It is well-known that this equation implies the Jacobi identity for~\eqref{app:cobound}. If the Lie algebra of a coboundary Poisson-Lie group $G$ is endowed with an invariant inner product, and the corresponding $r$-matrix satisfies the modified Yang-Baxter equation \eqref{app:mybe}, then $G$ is called \textit{factorizable}. In what follows, we will be interested in one particular type of $r$-matrices satisfying  the modified Yang-Baxter equation:
\begin{proposition}\label{app:gmt}
Let $\g$ be a Lie algebra endowed with an invariant inner product. Assume also that $\g$, as a vector space, can be written as a direct sum of three subalgebras $\g_\pos$, $\g_0$, and $\g_\neg$, such that $[\g_0, \g_\pos] \subset \g_\pos$, $[\g_0, \g_\neg] \subset \g_\neg$, the subalgebras $\g_\pos$, $\g_\neg$ are isotropic, and $\g_0$ is orthogonal to both $\g_\pos$ and $\g_\neg$.
 Then $r:= p_\pos - p_\neg$, where $p_\pos$, $p_\neg$ are projectors $\g \to \g_\pos$, $\g \to \g_\neg$ respectively, satisfies the modified Yang-Baxter equation, thus turning the group $G$ of the Lie algebra $\g$ into a factorizable Poisson-Lie group.

\end{proposition}
\begin{proof}
Direct verification of \eqref{app:mybe}.
\end{proof}

\begin{remark}\label{rem:gd}
Formula \eqref{app:cobound} for the coboundary Poisson-Lie bracket can be written in a more explicit form when the Lie group $G$ can be embedded, as an open subset, into an associative algebra $A$ (in a typical situation $G$ coincides with the group of invertible elements in $A$, e.g. the group of invertible matrices inside the algebra of all $n \times n$ matrices). In this case, the Lie algebra of $G$ and, more generally, the tangent space to $G$ at any point can be naturally identified with $A$. Assume also that $A$ is endowed with an invariant inner product, which in the context of associative algebras means that $\langle xy,z\rangle = \langle x,yz\rangle$ for any $x,y,z \in A$ (in particular, this inner product is invariant with respect to the adjoint action of $G \subset A$ on $A$). In that case the $r$-matrix can be thought of as a skew-symmetric operator $r \colon A \to A$, and identifying the cotangent space $T_g^*G$ with the tangent space $T_g G = A$ by means of the invariant inner product, one can rewrite formula~\eqref{app:cobound} for the corresponding Poisson tensor on $G$ as
\begin{equation}\label{app:gd}
\pi_g(x,y) = \frac{1}{2} \left(\vphantom{\int}\langle r(xg), yg\rangle - \langle r(gx), gy\rangle\right) \quad \forall\, g \in G, x,y \in A.
\end{equation}
The corresponding Poisson bracket is given by
\begin{equation}\label{app:gd2}
\{f_1, f_2\}(g) = \frac{1}{2} \left(\vphantom{\int}\langle r(\grad f_1(g)\cdot g), \grad f_2(g)\cdot g\rangle - \langle r(g\cdot \grad f_1(g)), g\cdot\grad f_2(g)\rangle\right),
\end{equation}
where the gradients are defined using the invariant inner product. Notice that the right-hand side of this formula is actually defined for every $g \in A$, i.e. invertibility of $g$ is not necessary. Therefore, this formula may be used to define a Poisson bracket on the whole of $A$. This bracket is known as the \textit{second Gelfand-Dickey bracket} on the associative algebra~$A$.
\end{remark}
In what follows we will need the following standard facts about coboundary Poisson-Lie groups:
\begin{proposition}\label{zeroPoisson}
Let $G$ be a Lie group endowed with a coboundary Poisson structure $\pi$ defined by $r$-matrix $\hat r$, and let $g \in G$. Then the Poisson structure $\pi$ vanishes at $g$ if and only if $(\mathrm{Ad}_g)_* \hat r = \hat r$.
\end{proposition}
\begin{proof}
We have $(\mathrm{Ad}_g)_*\hat r = (\rho_g^{-1})_* (\lambda_g)_*\hat r$, so $(\mathrm{Ad}_g)_* \hat r = \hat r$ if and only if $(\lambda_g)_*\hat r = (\rho_g)_*\hat r$, i.e. $\pi_g = 0$.
\end{proof}
\begin{proposition}\label{actionprop}
Let $\sigma \colon G \to G$ be an automorphism of a coboundary Poisson-Lie group. Assume that the differential of $\sigma$ at the identity preserves the $r$-matrix $\hat r$. Then $\sigma$ is a Poisson map.
\end{proposition}
\begin{proof}
Since $\sigma$ is an automorphism, we have
$
 \lambda_{\sigma(g)} = \sigma  \lambda_g \sigma^{-1}$ and   $\rho_{\sigma(g)} =  \sigma  \rho_g \sigma^{-1},
$
so 
$$
\pi_{\sigma(g)} = \frac{1}{2}\left(\vphantom{\int}(\lambda_{\sigma(g)})_*\hat r - (\rho_{\sigma(g)})_*\hat r \right)= \frac{1}{2}\left( \vphantom{\int}\sigma_* (\lambda_g)_* (\sigma^{-1})_*\hat r -  \sigma_* (\rho_g)_* (\sigma^{-1})_*\hat r\right).
$$
Since $\sigma$ preserves the $r$-matrix, the latter expression can be rewritten as
$$
 \frac{1}{2}\left(\vphantom{\int}  \sigma_* (\lambda_g)_* \hat r -  \sigma_* (\rho_g)_*\hat r \right) = \sigma_*\pi_g.
$$
So, $\pi_{\sigma(g)} = \sigma_*\pi_g$, which means that $\sigma$ is a Poisson map.
\end{proof}
\begin{proposition}\label{cf}
Central functions on a coboundary Poisson-Lie group Poisson commute.
\end{proposition}
\begin{proof}
Formula \eqref{app:cobound} is equivalent to
$$
\{f_1, f_2\}(g) =\frac{1}{2} \left(\vphantom{\int} \hat r(\lambda_g^* df_1(g),\lambda_g^* df_2(g) ) - \hat r(\rho_g^* df_1(g),\rho_g^* df_2(g) ) \right) \quad \forall\, f_1, f_2 \in C^\infty(G), g \in G.
$$
But for central functions $f_1$, $f_2$ we have $f_i \circ \lambda_g = f_i \circ \rho_g \Rightarrow \lambda_g^* df_i(g) = \rho_g^* df_i(g) \Rightarrow \{f_1, f_2\} = 0$.
\end{proof}

\subsection{Existence and properties of the Poisson structure}\label{app:mainproof}

In this section we prove Proposition \ref{app:mainprop} describing the Poisson structure on the group $\IPSIDO{n}$ of $n$-periodic invertible pseudo-difference operators.
To define that structure, we will use the construction described in Proposition \ref{app:gmt}. The Lie algebra of the group $\IPSIDO{n}$ is the space $\PSIDO{n}$ of all $n$-periodic  pseudo-difference operators. That is actually an associative algebra in which $\IPSIDO{n}$ is embedded as the set of invertible elements. 
That algebra has an invariant inner product defined by
\begin{equation}\label{app:ip}
\langle \D_1, \D_2\rangle = \Tr \D_1\D_2 \quad \forall \, \D_1, \D_2 \in \PSIDO{n},
\end{equation}
where the \textit{trace} of an $n$-periodic pseudo-difference operator $\D$ is given by
\begin{equation}\label{tracedef}
\Tr \left( \sum_{j=k}^\infty a_j T^i\right) := \sum_{i=1}^n a_{0,i}.
\end{equation}
The product \eqref{app:ip} is clearly non-degenerate and invariant in the associative algebra sense, i.e. $\langle \D_1, \D_2\D_3\rangle = \langle \D_1\D_2, \D_3\rangle.$ Furthermore, one can explicitly verify that $ \Tr \D_1\D_2 =  \Tr \D_2\D_1$, so the inner product \eqref{app:ip} is symmetric. Alternatively, this can be showed by using the isomorphism of $\PSIDO{n}$ and the algebra $\gl_n \otimes \R((z))$ of matrices over formal Laurent series (see Remark \ref{app:loopgroup}). In the matrix language, the trace of an operator can be written as
$
\Tr \D = \mathrm{Res}_{z=0}\left(z^{-1}\Tr \, \D(z)\right),
$
where $\D(z)$ is a matrix with coefficients in $\R((z))$ associated to the operator $\D$.
\begin{proof}[Proof of Proposition  \ref{app:mainprop}]
Represent the algebra $\PSIDO{n}$ of $n$-periodic pseudo-difference operators as the sum of three subalgebras $\g_\neg$, $\g_0$, $\g_\pos$ as follows. Let $\PSIDO{n}(J)$ be the vector space of pseudo-difference operators supported in $J \subset \Z$. By definition, an operator of the form  \eqref{app:psido} is supported in $J$ if $a_j \equiv 0$ for all $j \notin J$. Define
$$
\g_\neg := \PSIDO{n}(\Z_\neg) =  \DO{n}{\Z_\neg}, \quad \g_0   : = \PSIDO{n}(\{0\}) =   \DO{n}{\{0\}}, \quad \g_\pos := \PSIDO{n}(\Z_\pos),
$$
where $\Z_\pos$ stands for positive integers and $\Z_\neg$ for negative ones. 
 This decomposition clearly satisfies all the requirements of Proposition \ref{app:gmt}, so we get an $r$-matrix $r: = p_\pos - p_\neg$ and hence a factorizable Poisson-Lie structure on $\IPSIDO{n}$. This proves the first statement of Proposition \ref{app:mainprop}. To prove the second statement (invertible difference operators supported in a subset $J \subset \Z$ consisting of consecutive integers form a Poisson submanifold), we use formula~\eqref{app:gd}. From that formula it follows that, when viewed as map $\PSIDO{n} \to \PSIDO{n}$, the Poisson tensor $\pi_{\D}$ (where $\D \in \IPSIDO{n}$) reads
\begin{equation}\label{app:ho}
\pi_\D(\mathcal Q) = \D r(\mathcal Q\D) - r(\D \mathcal Q)\D.
\end{equation}
To show that the set $\IDO{n}{J} \subset \IPSIDO{n}$ of invertible difference operators supported in $J$ is a Poisson submanifold, one needs to prove that for $\D \in \IDO{n}{J}$ the image of the Poisson tensor \eqref{app:ho} belongs to the tangent space to $\IDO{n}{J}$ at $\D$. The latter is the space $\DO{n}{J}$ of all $n$-periodic difference operators supported in $J$, so we need to show that the operator \eqref{app:ho} is supported in $J$ whenever $\D$ is supported in $J$. To that end, notice that the right-hand side of  \eqref{app:ho} stays the same if $r$ is replaced by $r \pm \Id$. But the image of $r + \Id = 2 p_\pos + p_0$  (where $p_0$ is the projector to $\g_0$) is the space $\g_0 + \g_\pos$ of operators which only have terms of non-negative power in $T$, so,  rewriting \eqref{app:ho} in terms of $r + \Id$, we get that
$$
\min\mathrm{supp}\,\pi_\D(\mathcal Q) \geq \min \mathrm{supp}\,\D = \min J.
$$
Analogously, rewriting \eqref{app:ho} in terms of $r - \Id$, we get 
$
\max\mathrm{supp}\,\pi_\D(\mathcal Q) \leq  \max J.
$
All in all, we have
$
\mathrm{supp}\,\pi_\D(\mathcal Q) \subset [\min J, \max J] = J,
$
as desired.\par

To prove the third statement (if $J$ is one-point set, then the Poisson structure vanishes on operators supported in $J$), notice that if $\D$ is supported in a one-point set, then conjugation by $\D$ preserves the subalgebras $\g_\pm$ and $\g_0$, as well as the inner product on  $\PSIDO{n}$. Therefore, it preserves the $r$-matrix, and $\pi(\D) = 0$ by Proposition \ref{zeroPoisson}.\par

To prove the fourth statement (the left-right action is Poisson), we represent the left-right action~\eqref{action} as a superposition of two actions: one is of the same form, but with periodic $\alpha$ and $\beta$, while the other one is conjugation action $\D \mapsto \gamma \D \gamma^{-1}$, with quasi-periodic $\gamma$. Then the former action is Poisson because the Poisson structure vanishes on sequences, while the latter is Poisson because conjugation by sequences preserves $\g_\pos$, $\g_\neg$, and $\g_0$, as well as the inner product, and hence is Poisson by Proposition~\ref{actionprop}. So, the left-right action~\eqref{action} is also Poisson.
\par
Finally, the last statement of Proposition~\ref{app:mainprop} (central functions Poisson commute) directly follows from Proposition~\ref{cf}. So, Proposition~\ref{app:mainprop} is proved.
\end{proof}
\begin{remark}\label{rem:cf}
As central functions on $\IPSIDO{n}$, one can take expressions of the form
$
f_{ij}(\D) := \Tr {T^{in}\D^j}
$, where $i \in \Z, j \in \Z_\pos$. An alternative way to get the same functions is to consider the matrix-valued Laurent series $\D(z)$ corresponding to the operator $\D$ (cf. Remark \ref{rm:psiloop}), and then take coefficients in $z$ of the spectral invariants of $\D(z)$.
\end{remark}

\subsection{Relation to the $\GL_n$ bracket}\label{sec:gln}
One can compute Poisson brackets of coordinate functions on $\IPSIDO{n}$ using formula \eqref{app:gd2}. The resulting expressions are quite complicated and involve infinite series. However, only finitely many terms of those series are non-zero for every concrete pseudo-difference operator. Moreover, for a \textit{difference} operator whose support is small compared to the period these series simplify to just one term. Below we explain how to compute the brackets in this case by using the standard Poisson-Lie structure on $\GL_n$.\par
Recall that the standard Poisson structure on $\GL_n$ is defined using the construction of Proposition~\ref{app:gmt} with $\g_\neg$, $\g_\pos$, $\g_0$ being the lower nilpotent, upper nilpotent, and the Cartan subalgebra respectively, see e.g. \cite{gekhtman2009poisson}. Explicitly, the brackets of the matrix elements are given by
$$
\{x_{ij}, x_{kl} \} = \frac{1}{2}(\sgn(k-i) + \sgn(l-j))x_{il}x_{kj},
$$
where $\sgn(t)$ is $+1$ if $t> 0$,  $-1$ if $t> 0$, and  $0$ if $t = 0$. In other words, for any matrix entries $a,b,c,d$  located at vertices of a rectangle as shown below:
\begin{equation}
\left.\begin{array}{cccc}   a & \dots & \dots & b  \\  \vdots & &  & \vdots   \\ c & \dots & \dots& d  \end{array}\right.
\end{equation}
we have 
$$
\{a,b\} =\frac{1}{2}ab, \quad  \{a,c\} = \frac{1}{2}ac,  \quad \{a,d\} = bc, \quad  \{b,c\} =0.
$$
Since the relative position of $b$ and $d$ is the same as of $a$ and $c$, while the relative position of $c$ and $d$ is the same as of $a$ and $b$, we also have that $$\{b,d\} = \frac{1}{2}bd, \quad \{c,d\} = \frac{1}{2}cd. $$
We  now explain the relation between the bracket on difference operators and the $\GL_n$ bracket. Consider the algebra $\PSIDO{n}(\N)$, where $\N := \Z_\pos \cup \{0\}$, of $n$-periodic \textit{upper-triangular} pseudo-difference operators. Any such operator 
$
\D =  \sum_{j = 0}^{+\infty}  a_{j} T^j
$
can be represented by a bi-infinite upper-triangular matrix
$$
 \left.\begin{array}{ccccccccc}       & a_{0,i-1}   & a_{1,i-1} & \dots & \\ &  & a_{0,i} &  a_{1,i}  &\dots & &  \\ & &  & a_{0,i+1} & a_{1,i+1} & \dots &  \end{array}\right. .
 $$
 Let $\Phi_i(\D)$ be the $n \times n$ submatrix of this matrix which has the element $a_{0,i}$ in its upper left corner. 
 \begin{proposition}\label{windowLemma}
 Each of the mappings $\Phi_i \colon \PSIDO{n}(\N) \to \Mat_n$ takes the Poisson structure $\pi$ on $ \PSIDO{n}(\N) $ to the standard Poisson structure on $\Mat_n$.
 \end{proposition}
 \begin{remark}
 Technically, we have defined Poisson structures only on invertible pseudo-difference operators and invertible matrices. However, since both pseudo-difference operators and matrices form associative algebras, the Poisson structures in fact extend to non-invertible elements
  (see Remark~\ref{rem:gd}). 
 \end{remark}
  \begin{remark}
This proposition is saying that one can compute Poisson brackets of difference operator coefficients by sliding an $n \times n$ window through the operator matrix. If the support of the operator is not too big compared to the period, then the size of the window is big enough to fit any pair of the coefficients, so all Poisson brackets can be computed in this way.
 \end{remark}
 
 \begin{example}
  Consider the space of operators of the form $a + bT $. This corresponds to bi-infinite bi-diagonal matrices
 $$
 \left.\begin{array}{ccccccccc} \ddots & \ddots \\  & a_{i}   & b_{i}  & \\ & & a_{i+1}   & b_{i+1} & & \\ & & &    \ddots & \ddots  \end{array}\right..
 $$
 If $n = 1$, then we cannot use the $n \times n$ window to compute all the brackets. For $n \geq 2$, Proposition~\ref{windowLemma} gives
 \begin{align}\label{deg1bracket}
  \{a_i,b_i\} = \frac{1}{2}a_i b_i, \quad   \{b_i,a_{i+1}\} = \frac{1}{2}b_ia_{i+1},
\end{align}
while all other brackets are either obtained from these by shift of indices or vanish.
 \end{example}
 \begin{example}
 Consider the space of operators of the form $a + bT + cT^2$. The matrix of such an operator is
 $$
 \left.\begin{array}{ccccccccc}  \ddots   & \ddots& \ddots & & \\  & a_{i}   & b_{i} & c_{i} & \\ & & a_{i+1}   & b_{i+1} & c_{i+1}  & &  \\ & & & a_{i+2}   & b_{i+2} & c_{i+2}  &\\ & & &  &  \ddots & \ddots &\ddots \end{array}\right.
 $$
For $n \geq 3$, Proposition~\ref{windowLemma} gives
 \begin{align}\label{2ordflas}
\begin{aligned}
 \{a_i,&b_i\} = \frac{1}{2}a_i b_i, \quad \{a_i,c_i\} = \frac{1}{2}a_i c_i, \quad \{b_i, c_i\} = \frac{1}{2} b_ic_i, \quad  \{b_i,a_{i+1}\} = \frac{1}{2}b_ia_{i+1},   \\  &\{b_i, b_{i+1}\} = a_{i+1}c_i, \quad \{c_i, b_{i+1}\} = \frac{1}{2} c_i b_{i+1}, \quad
 \{c_i,a_{i+2}\} = \frac{1}{2}c_ia_{i+2}.
\end{aligned}
\end{align}
 \end{example}
 The proof of Proposition \ref{windowLemma} is based on the following lemma.
 \begin{lemma}\label{loopPoisson}
 Consider the space $\Mat_n \otimes \R[[z]]$ of formal matrix power series endowed with the trigonometric $r$-bracket, and the space  $\Mat_n$ endowed with the standard bracket. Then the mapping $$\Phi \colon \Mat_n \otimes \R[[z]] \to \Mat_n, \quad \Phi\left(\sum_{i=0}^\infty A_i z^i\right) := A_0,$$ taking a matrix power series to its constant term is a Poisson map.
 \end{lemma}
\begin{remark}
  The trigonometric $r$-bracket on the space $\Mat_n \otimes \R((z))$ of formal matrix Laurent series is defined using the construction of Proposition~\ref{app:gmt}, where $\g_\pos$ consists of matrix power series with nilpotent upper-triangular constant term, $\g_\neg$ consists of matrix polynomials in $z^{-1}$  with nilpotent lower-triangular constant term, while $\g_0$ is the space of constant diagonal matrices. The invariant inner product on $\Mat_n \otimes \R((z))$ is defined by $$\langle A(z), B(z) \rangle :=  \mathrm{Res}_{z=0}\left(\frac{1}{z}\Tr \, A(z)B(z)\right).$$
  \end{remark}
    \begin{proof}[Proof of Lemma \ref{loopPoisson}]
  The mapping $\Phi $ is well-defined on the whole space $\Mat_n \otimes \R((z))$ and maps both the $r$-matrix and the inner product on the latter space to the corresponding objects on  $\Mat_n$. Also notice that for any function $f \in C^\infty(\Mat_n)$, we have $\grad \Phi^* f \in \Mat_n \otimes \R[[z]]$. Indeed, the function $\Phi^*f$ is constant on the subspace 
  $$\Ker \Phi =  \left\{\sum_{i=1}^\infty A_i z^i\right\},$$
 so $ \grad \Phi^* f \in (\Ker \Phi)^\bot = \Mat_n \otimes \R[[z]]$. Furthermore, since $\Phi$ preserves the inner product, we have $\Phi(\grad (\Phi^* f)(A) ) = \grad f(\Phi(A))$. Now, take two functions $f_1, f_2 \in C^\infty(\Mat_n)$. Then the Poisson bracket of their $\Phi$-pullbacks at a point $A = A(z) \in \Mat_n \otimes \R[[z]]$ is given by formula \eqref{app:gd2}:
\begin{align*}
 \{\Phi^*f_1, \Phi^*f_2\}(A) &=  \frac{1}{2} \left(\vphantom{\int}\langle r(\grad \Phi^*f_1(A)\cdot A), \grad \Phi^*f_2(A)\cdot A\rangle \right. \\ &\qquad\qquad- \left. \vphantom{\int}\langle r(A\cdot \grad \Phi^*f_1(A)), A\cdot\grad \Phi^* f_2(A)\rangle\right).
\end{align*}
Using that both $A$ and the gradients of $f_1, f_2$ belong to $\Mat_n \otimes \R[[z]]$, while the restriction of $\Phi$ to $\Mat_n \otimes \R[[z]]$ is a homomorphism of associative algebras preserving the inner product and the $r$-matrix, this can be rewritten as
\begin{align*}
 \{\Phi^*f_1, \Phi^*f_2\}(A) &=  \frac{1}{2} \left(\vphantom{\int}\langle r(\grad f_1(\Phi(A))\cdot \Phi(A)), \grad f_2(\Phi(A))\cdot \Phi(A)\rangle \right. \\ &\qquad\qquad- \left. \vphantom{\int}\langle r(\Phi(A)\cdot \grad f_1(\Phi(A))), \phi(A)\cdot\grad f_2(\Phi(A))\rangle\right),
\end{align*}
which is exactly the $\Mat_n$ bracket of the functions $f_1, f_2$ at the point $\Phi(A)$. Thus, the mapping $\Phi$ is indeed Poisson, as claimed. 
  \end{proof}
 
 \begin{proof}[Proof of Proposition \ref{windowLemma}]
 In the loop algebra language, the space $\PSIDO{n}(\N)$ of upper-triangular $n$-periodic pseudo-difference operators is the space of formal matrix power series of the form  $A(z) = \sum_{i=0}^\infty A_i z^i,$ where $A_0$ is upper-triangular. The infinite matrix corresponding to such power series is
 $$
 \left.\begin{array}{cccc}       A_0   &A_1 & \dots &  \\  & A_0 &  A_1  &\dots    \end{array}\right.
 $$
 with upper-left corners of $A_0$ blocks located at positions $(jn+1, jn+1)$, $j \in \Z$. Thus, the mapping $\Phi_1$ takes $A(z)$ to $A_0$ and is, therefore, a restriction of the mapping $\Phi $ from Lemma \ref{loopPoisson}. So, $\Phi_1$ is a Poisson map. Furthermore, we have $\Phi_{i+1} = \Phi_i 
 \circ \Ad_T$, where $\Ad_T( \D) := T\D T^{-1}$. Therefore, since $\Ad_T$ is also a Poisson map (by Proposition \ref{app:mainprop}, item 3), it follows that all $\Phi_i$'s are Poisson, as desired.
 \end{proof}

\subsection{The subgroup of sparse operators}
We say that a pseudo-difference operator is \textit{$k$-sparse} if its support is an arithmetic progression with step $k$. For example, the operator $T^{-1} + T + T^3 + T^5 + \dots$ is $2$-sparse. Denote the set of invertible  {$k$-sparse} pseudo-difference operators by $\IPSIDO{n}(k\Z + \ast)$. This is a Lie subgroup of $\IPSIDO{n}$, whose Lie algebra is the space $\PSIDO{n}(k\Z)$ of  pseudo-difference operators supported in $k\Z$.  It is not, however, a Poisson submanifold and hence not a Poisson-Lie subgroup $\IPSIDO{n}$. One can, however, define a different Poisson structure on $\IPSIDO{n}(k\Z + \ast)$, which has all the same properties as the Poisson structure on $\IPSIDO{n}$ described above. More precisely, we have the following:
\begin{proposition}\label{prop:sparse}
There exists a natural Poisson structure $\pi^{(k)}$ on the group $\IPSIDO{n}(k\Z + \ast)$ of invertible  {$k$-sparse} pseudo-difference operators. It has all the same properties as the Poisson structure $\pi$ described in Proposition \ref{app:mainprop}, except for the second property which is replaced by the following: if $J \subset \Z$ is an arithmetic progression with common difference $k$, then $\IDO{n}{J}$ is a Poisson submanifold of  $\IPSIDO{n}(k\Z + \ast)$.
\end{proposition}
\begin{proof}
This Poisson structure is given by the following decomposition of the Lie algebra $\PSIDO{n}(k\Z)$:
$$
\g_\neg := \PSIDO{n}(k\Z_\neg), \quad \g_0   : = \PSIDO{n}(\{0\}), \quad \g_\pos := \PSIDO{n}(k\Z_\pos).
$$
All necessary properties are established in the same way as in the proof of Proposition \ref{app:mainprop}.
\end{proof}
\begin{remark}\label{rem:cbi}
A more constructive way to describe the Poisson structure $\pi^{(k)}$ is as follows. When $n$ and $k$ are coprime, there is a group isomorphism $\IPSIDO{n}(k\Z) \simeq \IPSIDO{n}$  given by the action of $\IPSIDO{n}(k\Z)$ on eventually vanishing sequences whose non-zero entries are contained in an arithmetic progression with common difference $k$. Explicitly, this isomorphism is given by
\begin{equation}\label{cbi}
\sum a_{kj} T^{kj} \mapsto \sum \tilde a_{j} T^{j},
\end{equation}
where $\tilde a_{j,i} = a_{kj, ki}$ (note that this is only an isomorphism when $n$ and $k$ are coprime; otherwise, this map is neither injective nor surjective). The Poisson structure $\pi^{(k)}$ can be defined as the pull-back of the structure $\pi$ by this isomorphism. Furthermore, $\pi^{(k)}$ uniquely extends to the whole group $\IPSIDO{n}(k\Z + \ast)$  if we require that the resulting structure is invariant under multiplication by $T$.  This gives a structure which coincides with the one described in the proof of Proposition \ref{prop:sparse}. 
Furthermore, this construction can also be applied when $n$ and $k$ are not coprime, in which case $\IPSIDO{n}(k\Z)$ is isomorphic to a product of $ m := gcd(n,k)$ copies of $\IPSIDO{n/m}$. The corresponding $m$ maps $\IPSIDO{n}(k\Z) \to \IPSIDO{n/m}$ are given by \eqref{cbi} with $\tilde a_{j,i} = a_{kj, ki + l}$, where $l = 0, \dots, m-1$.

\end{remark}
\begin{example}\label{ex:sparsebracket}
Consider sparse operators of the form $a + bT^2 $. The Poisson bracket $\pi^{(2)}$ on such operators may be obtained from the bracket $\pi$ on operators of the form $a + bT$ using the following mnemonic rule (justified by Remark \ref{rem:cbi}): take the formulas  \eqref{deg1bracket} for brackets on $a + bT$ and replace all indices of the form $i + j$ with $i + 2j$. This gives
 \begin{align}
   \{a_i,b_i\} = \frac{1}{2}a_i b_i, \quad   \{b_i,a_{i+2}\} = \frac{1}{2}b_ia_{i+2}.
\end{align}
\end{example}

\medskip
\section{General refactorization maps associated with pairs of arithmetic progressions}\label{sec3}

\subsection{The main theorem}\label{sec:mainThm}
In this section we describe a class of maps parametrized by pairs of finite arithmetic progressions $J_\pm \subset \Z$ with the same common difference. For disjoint $J_\pm$ these maps coincide with pentagram maps on $J$-corrugated polygons described in Section \ref{sec:geometry}. All these maps, regardless of whether  $J_\pm$ are disjoint, admit a refactorization description in terms of the group $\IPSIDO{n}$ of periodic pseudo-difference operators. As a corollary, all such maps admit an invariant Poisson structure and a  Lax representation with Poisson-commuting spectral invariants. Therefore, one should expect that all these maps are both Liouville and algebraically integrable. In order to actually prove that, one needs to accurately verify certain technical conditions, which is beyond the scope of the present paper.
\par
 Recall that the pentagram maps defined in Section \ref{sec:geometry} act on the space $\P_n(J) \,/\, \PGL$ of twisted $J$-corrugated $n$-gons modulo projective transformations, where $J:= J_+ \sqcup J_-$ is the union of two disjoint finite arithmetic progressions $J_\pm$ with the same common difference. By Proposition \ref{prop:quo}, that space can be identified with the space $\PBDO{n}{J} \, / \,H \ttimes H$ of properly bounded $n$-periodic difference operators supported in~$J$ modulo the left-right action \eqref{action} of the group $H \ttimes H$ of pairs of  non-vanishing $n$-quasi-periodic sequences with the same monodromy. Furthermore, decomposing a difference operator $\D \in \PBDO{n}{J}$ into a sum $\D_+ + \D_-$, where $\D_\pm \in \DO{n}{J_\pm}$ are difference operators supported in $J_\pm$, one can identify a dense subset in the quotient $\PBDO{n}{J} \, / \,H \ttimes H$ with $  \PBDO{n}{J_+} \times \PBDO{n}{J_-}   \, / \,H \ttimes H$, where $H \ttimes H$ acts on both factors by the \textit{simultaneous} left-right action (this identification is only possible for a dense subset of $\PBDO{n}{J} \, / \,H \ttimes H$ because the operators $\D_\pm$ may not be properly bounded even if $\D$ is). Thus, our pentagram maps, considered on sufficiently generic polygons, can be thought of as transformations defined on the left-right quotient $  \PBDO{n}{J_+} \times \PBDO{n}{J_-}  \, / \,H \ttimes H$, with $J_\pm$ being disjoint. Note, however, that the latter quotient is well-defined regardless of whether the sets $J_\pm$ are disjoint. 
 Below we describe certain dynamics on that quotient which in the disjoint case coincides with the pentagram dynamics.\par
 
 An alternative way to think of the space $  \PBDO{n}{J_+} \times \PBDO{n}{J_-}   \, / \,H \ttimes H$ is to identify it with the quotient $\PBDO{n}{J_-}^{-1}\PBDO{n}{J_+} \, /\, \Ad\,H$, where $\PBDO{n}{J_-}^{-1}\PBDO{n}{J_+}$ is the space of \textit{rational} pseudo-difference operators of the form  $ {\D}_-^{-1}{\D}_+$ with $\D_\pm \in  \PBDO{n}{J_\pm} $, and $\Ad\,H$ stands for the conjugation action of the group $H$ of $n$-periodic non-vanishing scalar sequences. The identification between the two spaces is done via the map $\D_\pm \mapsto {\D}_-^{-1}{\D}_+$, which we will show to be almost everywhere bijective. \par

 One last ingredient that we need to state the main result is a Poisson structure on our phase space  $  \PBDO{n}{J_+} \times \PBDO{n}{J_-}   \, / \,H \ttimes H$  (which, in the disjoint case, is the space of polygons). That structure is constructed as follows. Let $k$ be the common difference of $J_\pm$.
Define a Poisson structure on  $\PBDO{n}{J_+}$ as the restriction of the structure $\pi^{(k)}$ on $k$-sparse operators (see Proposition \ref{prop:sparse}). If $k=1$, that is just the standard structure on pseudo-difference operators (see Proposition  \ref{app:mainprop}). Further, on $\PBDO{n}{J_-}$, take the restriction of the same Poisson structure, but \textit{with an opposite sign}. This endows $ \PBDO{n}{J_+} \times \PBDO{n}{J_-} $ with a product Poisson structure. 
Furthermore, the quotient $ \PBDO{n}{J_+} \times \PBDO{n}{J_-}   \, / \,H \ttimes H$ inherits the Poisson structure because the left-right action is Poisson. 
%
%

The following theorem is the main result of the paper.
 \begin{theorem}\label{thm1}
 Let $J_\pm \subset \Z$ be a pair of non-empty finite arithmetic progressions with the same common difference. Consider the space $ \PBDO{n}{J_+} \times \PBDO{n}{J_-} $ of pairs $(\D_+, \D_-)$ of $n$-periodic properly bounded difference operators supported in $J_+, J_-$ respectively. Consider also the multivalued map of that space to itself that assigns to $\D_\pm$ new difference operators $\tilde \D_\pm$ defined by the equation
 \begin{equation}\label{mainRelation2}\tag{1'}
\tilde{\D}_+  {\D}_- = \tilde {\D}_-  {\D}_+.
\end{equation}
Then the following is true.
\begin{enumerate}
\item This map $\D_\pm \mapsto \tilde \D_\pm$ descends to a generically defined single-valued transformation $\Psi_{J_\pm}$ of the quotient $  \PBDO{n}{J_+} \times \PBDO{n}{J_-}   \, / \,H \ttimes H$, where $H \ttimes H$ is the group of pairs of  non-vanishing $n$-quasi-periodic sequences with the same monodromy acting by the simultaneous  left-right action~\eqref{action}.
\item If the progressions $J_\pm$ are disjoint, then the so-obtained map $\Psi_{J_\pm}$ coincides with the pentagram map associated with $J_\pm$.
\item The mapping 
\begin{equation}\label{eq:DivMap}
 \PBDO{n}{J_+} \times \PBDO{n}{J_-}   \, / \,H \ttimes H \to \PBDO{n}{J_-}^{-1}\PBDO{n}{J_+} \, /\, \Ad\,H.
\end{equation}
taking the left-right orbit of a pair $\D_\pm$ to the $H$-conjugacy class of the pseudo-difference operator ${\D}_-^{-1}{\D}_+$ is generically a bijection. This bijection identifies the map $\Psi_{J_\pm}$ with the following refactorization dynamics on conjugacy classes:
 \begin{equation}\label{refact}
\Lax := {\D}_-^{-1}{\D}_+ \, \mapsto \,  {\tilde \Lax} := {\D}_+{\D}_-^{-1}.
\end{equation}
In other words, the mapping $\Psi_{J_\pm}$ has a Lax representation
\begin{equation}\label{eq:lr}
\Lax \,\mapsto \, \D_+ \Lax {\D}_+^{-1}.
\end{equation}
\item The mapping $\Psi_{J_\pm}$ is Poisson. 
\item Suitably normalized central functions on the space of Lax operators $\Lax$ are Poisson commuting first integrals of $\Psi_{J_\pm}$.
\end{enumerate}
 \end{theorem}
 \begin{remark}
Here is what we mean by normalization of central functions. Recall that as central functions on the group $\IPSIDO{n}$ of $n$-periodic pseudo-difference operators one can take functions of the form $f_{ij}(\Lax) := \Tr {T^{in}\Lax^j}$, where $i \in \Z, j \in \Z_\pos$ (see Remark~\ref{rem:cf}). Upon conjugation of $\Lax$ by a quasi-periodic sequence $\alpha \in H$ with monodromy $z$, the function $f_{ij}$ transforms as
 $$
 f_{ij}(\alpha \Lax \alpha^{-1}) =  \Tr {T^{in}\alpha \Lax^j \alpha^{-1}} = z^i \Tr \alpha {T^{in} \Lax^j \alpha^{-1}} = z^i f_{ij}(\Lax).
 $$
 Thus, the functions $f_{ij}$ do not descend to the quotient of Lax operators by the conjugation action of $H$. One can, however, consider Laurent monomials of those functions that are invariant under the $H$-action and hence descend to the quotient. 
 \end{remark}
 \begin{remark}\label{specpar}
As explained in Remark \ref{rm:psiloop}, periodic pseudo-difference operators can be identified with matrices over the field $\R((z))$ of formal Laurent series. Thus, \eqref{eq:lr} can be viewed as a Lax representation valued in $\Mat_n \otimes \R((z))$, i.e. a Lax representation \textit{with spectral parameter} (in fact, since $\Lax$ is defined as a quotient of two difference operators, the corresponding matrix $\Lax(z)$ is not just a formal Laurent series but a rational function of $z$). Note, however, that since $\Lax$ is only defined up to conjugation by quasi-periodic sequences $\alpha \in H$, the corresponding $z$-dependent matrix $\Lax(z)$ is not uniquely defined. Namely, conjugation by periodic sequences translates to conjugation by $z$-independent diagonal matrices, while conjugation by a quasi-periodic sequence $\alpha_t$ defined by $\alpha_{t,i} := t^{\lfloor(i - 1) / n\rfloor}$ (where $\lfloor\, \rfloor$ is the floor function) becomes the action $\Lax(z) \mapsto \Lax(tz)$. Since $H$ is the direct product of periodic operators and the subgroup $\{ \alpha_t \mid t \in \R^*\}$, it follows that the Lax matrix $\Lax(z)$ is defined up to transformations of the form $\Lax(z) \mapsto A\Lax(tz)A^{-1}$, where $A$ is a constant invertible diagonal matrix. 
 \end{remark}
 \begin{example}
 For the classical pentagram map, the matrix Lax representation is given by
\begin{gather}
 \D_+ = \left(\begin{array}{cccccc}a_1 & 0 & c_1 &  &  &  \\ & \ddots & \ddots & \ddots &  &  \\ &  & \ddots & \ddots & \ddots &  \\  &  &  & a_{n-2} & 0 & c_{n-2} \\  c_{n-1}z & &  &  & a_{n-1} & 0 \\  0 & c_nz & &  &  & a_n\end{array}\right), \end{gather} \begin{gather}\D_-=  \left(\begin{array}{cccccc}0 & b_1 & 0 & d_1 &  &  \\ & \ddots & \ddots & \ddots & \ddots  &  \\ &  & 0 & b_{n-3} & 0 & d_{n-3} \\  d_{n-2}z&  &  & 0 & b_{n-2} & 0 \\  0 & d_{n-1}z &  &  & 0 & b_{n-1} \\  b_nz & 0 & d_nz &  &  & 0\end{array}\right),\\
 \mathcal L:= {\D}_-^{-1}{\D}_+, \quad  \tilde \Lax = \D_+ \Lax {\D}_+^{-1},
\end{gather}
where  $\mathcal L$ is determined by a polygon up to conjugation by constant diagonal matrices and rescaling $z \mapsto tz$, and $\D_\pm$ are determined by $\mathcal L$ up to simultaneous left multiplication by diagonal matrices.
 \end{example}
 One can also characterize first integrals of the maps $\Psi_{J_\pm}$ provided by Theorem \ref{thm1} as follows:
  \begin{corollary}\label{cor:int}
 The characteristic polynomial $\CPM{\D_+ + w\D_-}{z}$ of the monodromy of $\D_+ + w\D_-$, defined up to transformations of the form $z \mapsto tz$ and a constant factor, is invariant under the map~$\Psi_{J_\pm}$.
 \end{corollary}
 \begin{proof}[Proof of Corollary \ref{cor:int}]
Let $\D_\pm(z), \Lax(z)$ be the loop group elements corresponding to the operators $\D_\pm$ and $\Lax$ respectively. Then, by Theorem \ref{thm1}, the map $\Psi_{J_\pm}$ preserves the central function
\begin{equation}\label{fracIntegral} \det( \Lax(z) + w\Id) = \det( {\D}_-^{-1}(z){\D}_+(z)  + w\Id) =\frac{\det(\D_+(z) + w\D_-(z) )}{\det \D_-(z) }\end{equation} defined up to transformations of the form $z \mapsto tz$, cf. Remark \ref{specpar}. Using Proposition \ref{prop:detmono}, we can further rewrite this as

\begin{equation}\label{fracIntegral2} \det( \Lax(z) + w\Id) = z^{k}  \frac{\CPM{\D_+ + w\D_-}{z}}{\CPM{\D_-}{z} }, \quad k: = \min(J_- \cup J_+) - \min J_-.\end{equation} 
The fraction in the right-hand side is generically irreducible, so both its numerator and denominator must be preserved by $\Psi_{J_\pm}$, up to a constant factor.
 \end{proof}
 
 \begin{remark}
 In Section \ref{sec:scaling} we use Corollary \ref{cor:int} to show that in the known cases of integrability our first integrals coincide with the known ones. Furthermore, one can show that our Poisson structures also coincide with the familiar ones in those cases where a Poisson structure was previously known, namely for the classical pentagram map, leapfrog map (see Remark \ref{rm:leapfrog} below), as well as for pentagram maps on corrugated polygons. For short-diagonal and dented maps no invariant Poisson structures were previously known. In Section \ref{sec:poisson} we derive, as an example, explicit formulas for the Poisson structure of the short-diagonal map in 3D.
 \end{remark}

 \begin{remark}\label{rm:leapfrog}
 For $J_+ \cap J_- \neq \emptyset$, the geometric meaning of the maps  $\Psi_{J_\pm}$ is not known. The only case which we were able to identify with a familiar integrable system is $J_+ = \{-1,0\}$, $J_- = \{0,1\}$ (and, more generally, $J_+ = \{k-1,k\}$, $J_- = \{k,k+1\}$ which correspond  to the same map up to a shift indices). In that case, the map $\Psi_{J_\pm}$ is the leapfrog map of \cite{Gekhtman2016}, also known as the discrete relativistic Toda lattice~\cite{suris1997some}. The phase space of the leapfrog map is, by definition, the space of pairs of twisted $n$-gons in $\RP^1$ with the same monodromy, considered up to simultaneous projective transformations. One can lift such two polygons to two bi-infinite sequences $V_i^-, V_i$ of vectors in $\R^2$, and then construct two operators $\D_-$ and $\D_+$ supported in $J_+ = \{-1,0\}$ and $J_- = \{0,1\}$ respectively such that $(\D_- + \D_+)V^- = 0$ and $\D_-V^- = V.$ This identifies  the space of pairs of twisted $n$-gons in $\RP^1$ with the same monodromy, considered up to simultaneous projective transformations, with the left-right quotient $  \PBDO{n}{J_+} \times \PBDO{n}{J_-}    \, / \,H \ttimes H$, while the leapfrog map gets identified with the map $\Psi_{J_\pm}$. The only proof of this we were able to find consists of expressing both maps in coordinates.  However, we do believe that it should be possible to directly identify equation \eqref{mainRelation2} with the geometric ``leapfrogging'' definition, similarly to how we identify it with pentagram-type dynamics in the case of disjoint $J_\pm$.
 \end{remark}
 
 To prove Theorem \ref{thm1} we essentially repeat the argument we used to prove Theorem \ref{thm0}, filling in technical details. We begin with a few lemmas.
 
 \begin{lemma}\label{l1}
 Let $\D$, $\D'$ be $n$-periodic difference operators with the same support, and let $\D$ be properly bounded. Assume that $\Ker \D' \supset \Ker \D$. Then there is an $n$-periodic sequence $\alpha$ such that $\D' = \alpha \D$. 
 \end{lemma}
 \begin{proof}
Without loss of generality, assume that $\D$ and $\D'$ are supported in $\{0,\dots, d\}$. Let $\alpha$ be the leading coefficient of $\D'$ divided by the leading coefficient of $\D$. Then the difference operator $ \mathcal R := \D' - \alpha \D$ is supported in $\{0,\dots, d-1\}$ and annihilates the kernel of $\D$. Let us show that such an operator must be zero. Assume $\mathcal R \neq 0$. Then there is $\xi \in \R^\infty$ and $i \in \Z$ such that $(\mathcal R\xi)_i \neq 0$. Further, since $\D$ is properly bounded of degree $d$, there is $\hat \xi \in \Ker \D$ such that $\hat \xi_j = \xi_j$ for all $j \in \{i, \dots, i+d-1\}$. Then, on one hand, since $\Ker \mathcal R \supset \Ker \D$, we have $\mathcal R \hat \xi = 0$. On the other hand, since $\hat \xi_j = \xi_j$ for all $j \in \{i, \dots, i+d-1\}$ and $\mathcal R$ is supported in $\{0,\dots, d-1\}$, we have that $(\mathcal R\hat \xi)_i = (\mathcal R \xi)_i \neq 0$. So we indeed must have $\mathcal R = 0$ and $\D' = \alpha \D$, as desired.
 \end{proof}
  \begin{lemma}\label{l2}
Let $\D_\pm \in \PBDO{n}{J_\pm}$ be operators with trivially intersecting kernels. Assume also that $\tilde \D_\pm \in \PBDO{n}{J_\pm}$ are properly bounded operators satisfying  \eqref{mainRelation2}. Then, for any other operators $\tilde \D_\pm' \in \DO{n}{J_\pm}$ satisfying  \eqref{mainRelation2}, there exists an $n$-periodic sequence $\alpha$ such that $\tilde \D_\pm' = \alpha \tilde \D_\pm$.

 \end{lemma}
\begin{proof}
First assume that $\tilde \D_\pm \in \DO{n}{J_\pm}$ is any solution of  \eqref{mainRelation2}. 
Then, applying both sides of  \eqref{mainRelation2} to any $\xi \in \Ker \D_+$ we get  $\tilde{\D}_+  {\D}_-\xi = 0$, meaning that \begin{equation}\label{eq:dimineq}  \Ker \tilde \D_+ \supset \D_-(\Ker \D_+).\end{equation} 
Now assume that  $\tilde \D_+$ is properly bounded. Then $\dim \Ker \tilde \D_+ = \dim \Ker  \D_+$. At the same time, since $\Ker \D_+ \cap \Ker \D_- = 0$, it follows that $\dim \D_-(\Ker \D_+) = \dim \Ker \D_+$.
So, the dimensions of both sides of \eqref{eq:dimineq} are the same. Thus, for any solution of  \eqref{mainRelation2} we have inclusion \eqref{eq:dimineq} while for properly bounded solutions the inclusion becomes an equality.
That means that if $\tilde \D_\pm$ is a properly bounded solution, and $\tilde \D_\pm' $ is any other solution, then $ \Ker \tilde \D_+' \supset  \Ker \tilde \D_+$. So, by Lemma \ref{l1} we have $\tilde \D_+' = \alpha \tilde \D_+$ for some periodic sequence $\alpha$. 
%
But then, using that both pairs $\tilde \D_\pm$, $\tilde \D_\pm' $ solve  \eqref{mainRelation2}, we get
$$
\tilde \D_-' = \tilde{\D}_+'  {\D}_-\D_+^{-1} = \alpha\tilde{\D}_+ {\D}_-\D_+^{-1} = \alpha \tilde \D_-.
$$
So,  $\tilde \D_\pm' = \alpha \tilde \D_\pm$, as desired. 
\end{proof}
\begin{lemma}\label{l3}

There exists a Zariski open and dense subset $\mathcal A({J_\pm} )\subset \PBDO{n}{J_+} \times \PBDO{n}{J_-} $ such that for any $(\D_+, \D_-) \in\mathcal A({J_\pm} )$ equation \eqref{mainRelation2} admits a solution $(\tilde \D_+, \tilde \D_-) \in  \PBDO{n}{J_+} \times \PBDO{n}{J_-}$. Moreover, this solution is unique up to multiplying both $\tilde \D_+$ and $\tilde \D_-$ by the same periodic sequence on the left.
\end{lemma}
\begin{proof}

 Let ${\mathcal A'} \subset  \PBDO{n}{J_+} \times \PBDO{n}{J_-}$ be the set of  pairs $\D_\pm$  such that \eqref{mainRelation2} has a unique solution $(\tilde \D_+, \tilde \D_-) \in  \DO{n}{J_+} \times \DO{n}{J_-}$  with monic  $\tilde \D_+$ (which means that the leading coefficient of $\tilde \D_+$ is equal to $1$).  This set is Zariski open. Indeed,  solving \eqref{mainRelation2} for  $\tilde \D_\pm$ with $\tilde \D_+$ monic is equivalent to a linear system. The number of indeterminates in that system is the number of unknown coefficients of $\tilde \D_\pm$ multiplied by the period, that is $n\left(|J_+| + |J_-| - 1\right)$. At the same time, since both sides of~\eqref{mainRelation2} are operators supported in the Minkowski sum $J_+ + J_-$, the number of equations is $n|J_+ + J_-|$, which, for two arithmetic progressions with the same common difference is also equal to $n\left(|J_+| + |J_-| - 1\right)$. So, we have a linear system where the number of unknowns is the same as the number of equations, and uniqueness of the solution is equivalent to non-vanishing of the determinant. That determinant is a polynomial in coefficients of $\D_\pm$, so the set ${\mathcal A'} $ is Zariski open. \par
 Now, define $\mathcal A \subset {\mathcal A'} $ as the set of those pairs $\D_\pm$ that belong to ${\mathcal A'} $ and have the property that the unique solution  $(\tilde \D_+, \tilde \D_-)$ of  \eqref{mainRelation2}  with monic $\tilde \D_+$ belongs to $\PBDO{n}{J_+} \times \PBDO{n}{J_-}$. 
This set is also Zariski open. Indeed, as we just saw, the unique solution of  \eqref{mainRelation2}  with monic $\tilde \D_+$ comes from an $m \times m$ linear system with coefficients being polynomials in $\D_\pm$. So, the solution is properly bounded when certain rational functions do not vanish, which means $\mathcal A$ is Zariski open in $ {\mathcal A'}$ and hence in $ \PBDO{n}{J_+} \times \PBDO{n}{J_-} $.\par
We now show that the set $\mathcal A$ is not empty. To that end, assume that  $\D_\pm  \in \PBDO{n}{J_\pm}  $ are operators with constant coefficients such that $\Ker \D_+ \cap \Ker \D_- = 0$ and $\D_+$ is monic. Then, since $\D_\pm$ have constant coefficients, they commute with each other, which is equivalent to saying that $\tilde \D_\pm := \D_\pm$ solve~\eqref{mainRelation2}. That solution has monic  $\tilde \D_+$ and, moreover, there are no other monic solutions. Indeed, by Lemma~\ref{l2} any other solution must be of the form $\tilde \D_\pm = \alpha \D_\pm$,  so if   $\tilde \D_+$ is monic then $\alpha = 1$ and $\tilde \D_\pm = \D_\pm$. Therefore, for $\D_\pm$ as described, equation ~\eqref{mainRelation2} admits  a unique solution with $\D_+$ is monic, and that solution is properly bounded. But that means $(\D_+, \D_-) \in \mathcal A$, as desired. \par
So, the set $\mathcal A$ is Zariski open and non-empty, hence open and dense. Now, to complete the proof, it suffices to show that for any $(\D_+, \D_-) \in\mathcal A({J_\pm} )$ the solution $(\tilde \D_+, \tilde \D_-) \in  \PBDO{n}{J_+} \times \PBDO{n}{J_-}$ of equation \eqref{mainRelation2} is unique up to multiplying both $\tilde \D_\pm$ by the same periodic sequence on the left. Indeed, if there were two solutions not related in this way, then dividing them by the leading term of $\tilde \D_+$ we would obtain two different solutions with monic $\tilde \D_+$. But that is not possible by construction of $\mathcal A$. Thus, the set $\mathcal A({J_\pm} ):= \mathcal A$ satisfies our requirements. 
\end{proof}
\begin{lemma}\label{l4}

There exists a Zariski open and dense subset $\mathcal B({J_\pm} ) \subset \PBDO{n}{J_+} \times \PBDO{n}{J_-} $ such that for any $(\D_+, \D_-) \in \mathcal B({J_\pm} )$ the pseudo-difference operator $ {\D}_-^{-1}{\D}_+$ has a unique representation as a left quotient of operators supported in $J_\pm$, up to multiplying both $\D_+$ and $\D_-$ by a periodic sequence on the left.

\end{lemma}
\begin{proof}
Recall that \textit{duality} $\D \mapsto \D^*$ is an anti-automorphism of the algebra of difference operators that is uniquely determined by requiring that a scalar sequence is self-dual, and $T^* = T^{-1}$. In other words, the dual of an operator $\D = \sum a_i T^i$ is $$\D^* = \sum T^{-i}a_i = \sum \tilde a_iT^{-i}, \quad \tilde a_{i,j} = a_{i, j - i}. $$ 
This corresponds to transposition of the operator matrix and can be thought of as operator duality with respect to the formal $L^2$ inner product on $\R^\infty$. Also note that the dual of a properly bounded operator supported in $J \subset \Z$ is a properly bounded operator supported in $-J := \{-j \mid j\in J\}$. \par
Let  $\star \colon \PBDO{n}{J_+} \times \PBDO{n}{J_-} \to  \PBDO{n}{-J_+} \times \PBDO{n}{-J_-} $ be the map that takes $(\D_+, \D_-)$ to $(\D_+^*, \D_-^*)$. Let also $\psi_{J_\pm} \colon \mathcal A({J_\pm} )\to \PBDO{n}{J_+} \times \PBDO{n}{J_-} $ be the map which takes a pair $(\D_+, \D_-)$ to the unique monic properly bounded solution $(\tilde \D_+, \tilde \D_-)$ of equation~\eqref{mainRelation2}. As we saw in the proof of Lemma~\ref{l3}, that is a rational map. Consider now the map
$$
\psi_{J_\pm}^* := \star \circ \psi_{-J_\pm} \circ \star \colon \mathcal A({-J_\pm} )^* \to  \PBDO{n}{J_+} \times \PBDO{n}{J_-},
$$
where $\mathcal A({-J_\pm} )^* := \star(\mathcal A({-J_\pm} ))$, and let $\mathcal B := (\psi_{J_\pm}^*)^{-1}(\mathcal A_{J_\pm})$. Then $\mathcal B$ is Zariski open as the preimage of a Zariski open set under a rational map. Furthermore, $\mathcal B$ is non-empty. Indeed, let $\D_\pm \in  \PBDO{n}{-J_\pm} $ be operators with constant coefficients such that $\Ker \D_+ \cap \Ker \D_- = 0$ and $\D_+$ is monic. Then $(\D_+, \D_-) \in \mathcal A({-J_\pm} )$, so $(\D_+^*, \D_-^*) \in  \mathcal A({-J_\pm} )^*$. Furthermore, we have
$$
\psi_{J_\pm}^* (\D_+^*, \D_-^*) = (\D_+^*, \D_-^*),
$$
so $(\D_+^*, \D_-^*) \in \mathcal B$. Thus, the set $\mathcal B({J_\pm} ) := \mathcal B$ is Zariski open and non-empty, and to complete the proof it suffices to show it satisfies the unique factorization requirement. To that end, assume that $(\D_+, \D_-) \in \mathcal B$, and ${\D}_-^{-1}{\D}_+ = ({\D}_-')^{-1}{\D}_+'$
 for some operators $\D_\pm' \in \PBDO{n}{J_\pm}$. Since  $(\D_+, \D_-) \in \mathcal B$, we have that  $(\D_+^*, \D_-^*) \in \mathcal A({-J_\pm})$, meaning there exists a unique pair of operators $\tilde \D_\pm \in \PBDO{n}{-J_\pm}$ with $\tilde \D_+$ monic such that
\begin{equation}\label{eq:predual}
\tilde \D_+ \D_-^* = \tilde \D_- \D_+^*.
\end{equation}
Moreover, by definition of  $\mathcal B$  we have $(\tilde \D_+^*, \tilde \D_-^*) \in\mathcal A({J_\pm} )$. Now, taking the dual of \eqref{eq:predual} we get
\begin{equation}\label{eq:dual}
 \D_- \tilde \D_+^* =  \D_+ \tilde \D_-^* \, \Rightarrow \,\tilde \D_+^*(\tilde \D_-^*)^{-1} =  \D_-^{-1}\D_+ = ({\D}_-')^{-1}{\D}_+' \, \Rightarrow \, \D_-' \tilde \D_+^* =  \D_+' \tilde \D_-^*.
\end{equation}
So, since $(\tilde \D_+^*, \tilde \D_-^*) \in\mathcal A({J_\pm} )$, by Lemma \ref{l3} we have $\D'_\pm = \alpha \D_\pm$ for some periodic sequence $\alpha$, as required.
\end{proof}
 
 \begin{proof}[Proof of Theorem \ref{thm1}] We begin with the first statement of the theorem (equation \eqref{mainRelation2} defines a generically single-valued map $\Psi_{J_\pm}$ of the left-right quotient to itself). By Lemma \ref{l3}, for generic $\D_\pm  \in \PBDO{n}{J_\pm}  $ equation \eqref{mainRelation2} has a solution  $\tilde \D_\pm  \in \PBDO{n}{J_\pm}  $ which is unique up to multiplying both $\tilde \D_\pm$ on the left by some $n$-periodic sequence $\alpha$. Thus, that equation defines a generically defined and generically single-valued map from the space
 $
  \PBDO{n}{J_+} \times \PBDO{n}{J_-} 
 $
 to its left quotient by the group $H_0 := \IDO{n}{\{0\}} \subset H$ of non-vanishing $n$-periodic sequences. To show that this map descends to the left-right quotient, it suffices to check that if the preimages are in the same left-right orbit, then so are the images (note that the left-right action is still defined on the left quotient by $H_0$, although it is not faithful). Assume that \eqref{mainRelation2} takes a pair $\D_\pm$ to the $H_0$-orbit of $\tilde \D_\pm$. Take another element of the left-right orbit of $\D_\pm$. That has the form $\alpha \D_\pm \beta^{-1}$ for some quasi-periodic sequences $\alpha$, $\beta$ with the same monodromy. Then  \eqref{mainRelation2} has a solution given by  $\beta \tilde \D_\pm \alpha^{-1}$. So, indeed elements of the same left-right orbit are mapped to elements of the same left-right orbit, proving the first statement of the theorem.

The proof of the second statement (for disjoint $J_\pm$ the maps $\Psi_{J_\pm}$ coincide with pentagram maps on $J$-corrugated polygons) repeats, word for word, the proof of the corresponding part of Theorem \ref{thm0}, so we proceed to the third statement (the map $\Psi_{J_\pm}$ can be identified with refactorization dynamics on rational operators). First, we need to show that the map \eqref{eq:DivMap} given by  $\D_\pm \mapsto {\D}_-^{-1}{\D}_+$ is generically a bijection. It is clearly surjective by definition of the codomain, so it suffices to prove injectivity. That is, we need to show that if  $ {\D}_-^{-1}{\D}_+$ is $H$-conjugate to  $({\D'}_-)^{-1}{\D}_+'$, then for generic $\D_\pm$ the pairs $\D_\pm$ and $\D_\pm'$ are in the same left-right orbit. To that end, assume that
$$
{\D}_-^{-1}{\D}_+ = \alpha^{-1}({\D}_-')^{-1}{\D}_+'\alpha = ({\D}_-'\alpha)^{-1}{\D}_+'\alpha
$$
for some periodic sequence $\alpha \in H$. Then, for generic $\D_\pm$, by Lemma \ref{l4} we have $\D_\pm' \alpha = \beta \D_\pm$. But that precisely means that the pairs $\D_\pm$ and $\D_\pm'$ are in the same left-right orbit, as desired. \par
Now that we know that  \eqref{eq:DivMap} is a bijection, we show that it identifies  $\Psi_{J_\pm}$ with refactorization dynamics. Indeed, \eqref{mainRelation2} is equivalent to \begin{equation}\label{refact2}\tag{2'} \tilde {\D}_- ^{-1}\tilde{\D}_+  =  {\D}_+ {\D}_-^{-1}, \end{equation} which precisely means that the operator 
$\tilde \Lax :=  {\tilde \D}_-^{-1}{\tilde \D}_+$ associated with $\tilde \D_\pm$ is obtained from the operator $\Lax := {\D}_-^{-1}{\D}_+$ associated with $ \D_\pm$ by means of refactorization { \eqref{refact}}.

To prove the fourth statement (the mapping $\Psi_{J_\pm}$ is Poisson), depict \eqref{refact2} as the following commutative diagram:
\def\clap#1{\hbox to 0pt{\hss#1\hss}}
\def\mathllap{\mathpalette\mathllapinternal}
\def\mathrlap{\mathpalette\mathrlapinternal}
\def\mathclap{\mathpalette\mathclapinternal}
\def\mathllapinternal#1#2{%
\llap{$\mathsurround=0pt#1{#2}$}}
\def\mathrlapinternal#1#2{%
\rlap{$\mathsurround=0pt#1{#2}$}}
\def\mathclapinternal#1#2{%
\clap{$\mathsurround=0pt#1{#2}$}}
 \begin{equation}\label{diagramma}
\centering
\begin{tikzcd}[ row sep = huge, column sep = tiny]
 \PBDO{n}{J_+} \times \PBDO{n}{J_-}    \, / \,H \ttimes H \arrow[swap,dr,end anchor={[yshift=1ex]},"{{\D}_+{\D}_-^{-1}}"]\arrow{rr}{\Psi_{J_\pm}}& &  \PBDO{n}{J_+} \times \PBDO{n}{J_-}    \, / \,H \ttimes H\arrow[dl,end anchor={[yshift=1ex]},"{{\D}_-^{-1}{\D}_+}"] \\
&   \hphantom{\quad} \clap{$\PBDO{n}{J_-}^{-1}\PBDO{n}{J_+} \, /\, \Ad\,H$} \hphantom{\quad}&
\end{tikzcd}
\end{equation}
Note that the left diagonal arrow is well-defined because by Lemma \ref{l3} almost every right quotient ${\D}_+{\D}_-^{-1}$ can be rewritten as a left quotient  ${\tilde \D}_-^{-1}{\tilde \D}_+$, so $$\PBDO{n}{J_+}\PBDO{n}{J_-}^{-1} = \PBDO{n}{J_-}^{-1}\PBDO{n}{J_+},$$ up to Zariski closed subsets (Lemma \ref{l3} is only one containment direction, while the opposite one can be proved by applying the lemma to dual operators, as in Lemma~\ref{l4}). 
Furthermore, the diagonal arrows are Poisson, since multiplication in $\IPSIDO{n}$ is Poisson, inversion is anti-Poisson, and the Poisson structure on the space $ \PBDO{n}{J_+} \times \PBDO{n}{J_-} $ of pairs of operators is defined by reversing the structure on the factor corresponding to $\D_-$. Also notice that by item 3 the right diagonal arrow is generically invertible. So, $\Psi_{J_\pm}$ is a composition of Poisson maps and hence Poisson, as stated.\par
Finally, we prove the fifth statement (central functions of $\Lax$ are Poisson-commuting first integrals of the map $\Psi_{J_\pm}$). Central functions on $\IPSIDO{n}$ applied to $\Lax$ are preserved by the map $\Psi_{J_\pm}$ due to representation \eqref{eq:lr} so it suffices to prove that they commute.  More precisely, we need to establish Poisson commutativity for the pull-backs of central functions on  $\IPSIDO{n}$ by the map~\eqref{eq:DivMap}. But that follows from commutativity of central functions on $\IPSIDO{n}$ along with the fact that \eqref{eq:DivMap} is a Poisson map (proved in item 4). So, Theorem \ref{thm1} is proved.
 \end{proof}

\subsection{Scaling invariance}\label{sec:scaling}
Most of the known constructions of first integrals and Lax representations for pentagram-type maps are based on \textit{scaling symmetries}. A scaling symmetry is a $1$-parametric group of transformations of the polygon space which commutes with the pentagram map. In most cases such symmetries were guessed by studying explicit formulas for the corresponding map, and their geometric meaning is not known. The aim of this section is to show that the scaling symmetry is an immediate corollary of our construction.

%
\begin{proposition}\label{prop:scaleDef}
The map $\Psi_{J_\pm}$, described in Theorem \ref{thm1}, commutes with a $1$-parametric group $R_w$ of transformations which is defined, in terms of difference operators, as
\begin{equation}\label{scaleDef}
 \D_+, \D_- \mapsto \D_+, w \D_-.
\end{equation}
In terms of the Lax operator, this transformation is simply rescaling:
\begin{equation}\label{scaleDef2}
\Lax(z) \mapsto w^{-1}\Lax(z).
\end{equation}
\end{proposition}
\begin{remark}
Transformation \eqref{scaleDef} commutes with the left-right $H \ttimes H$ action \eqref{action} (while \eqref{scaleDef2} commutes with the conjugation action) and hence can be viewed as a map from the space $ \PBDO{n}{J_+} \times \PBDO{n}{J_-}    \, / \,H \ttimes H$ (which is where the map $\Psi_{J_\pm}$ is defined) to itself. 
\end{remark}
\begin{proof}[Proof of Proposition \ref{prop:scaleDef}]
Indeed, the defining equation \eqref{mainRelation2} of the map $\Psi_{J_\pm}$ is invariant under the transformation $\D_- \mapsto w\D_-$, $\tilde \D_- \mapsto w\tilde \D_-$, while the Lax form  \eqref{eq:lr} is invariant under rescaling. \end{proof}
\begin{proposition}
In the case of the classical pentagram map, as well as in short-diagonal and dented cases, transformations $R_w$ defined in Proposition \ref{prop:scaleDef} coincide with scaling transformations introduced for these maps in \cite{schwartz2008discrete, khesin2013, khesin2016}.
\end{proposition}
\begin{proof}
The proof is achieved by introducing coordinates on the polygon space and rewriting the scaling symmetry in those coordinates. As an example, let us consider short-diagonal maps in $\RP^{2k}$ (the proof in other cases is analogous). This corresponds to $J_+ = \{0,2,4, \dots, 2k\} $, $J_- = \{1,3,5, \dots, 2k+1\}$ (see Table \ref{tab:pentmaps}). The phase space of the associated short-diagonal map is the space $\P_n(J) \, / \, \PGL$, with $J = J_+ \sqcup J_- = \{0,\dots, 2k+1\}$, of arbitrary (twisted) polygons in $\RP^{2k}$,  modulo projective equivalence. In terms of difference operators, it is the space of operators supported in $J$ and considered modulo the left-right action  \eqref{action}  of $H \ttimes H$. As can be seen from \cite[Section 3.2]{khesin2013}, as well as from \cite[Section 8.2]{morier2014linear}, if $gcd(2k+1, n)= 1$, then every orbit of the $H \ttimes H$ action has a unique representative of the form
\begin{equation}\label{normdo}\D = 1 + \sum_{j = 1}^{2k}a_j T^j - T^{2k+1}.\end{equation}
Thus, one can take entries of the sequences $a_j, j = 1, \dots, 2k$, as coordinates on the polygon space. To write our scaling transformation $R_w$ in these coordinates, we need to apply it to operator \eqref{normdo}, which gives
\begin{equation}\label{normdoscaled} \D' = 1 + \sum_{j = 1}^{2k-1} w a_j T^j + \sum_{j = 2}^{2k}  a_j T^j  - wT^{2k+1},\end{equation}
and then normalize, i.e. find an operator  $\tilde \D$ of the form \eqref{normdo} which belongs to the same orbit of the $H \ttimes H$ action as~\eqref{normdoscaled}. Note that since the constant term of $ \D'$ is already of necessary form, it remains to normalize the coefficient of $T^{2k+1}$, which can be done using only the conjugation action of $H$. The condition for $\alpha  \D' \alpha^{-1}$, where $\alpha \in H$, to have coefficient of $T^{2k+1}$ equal to $-1$ is
$
{\alpha_{i+2k+1}}{\alpha_i^{-1}} = w.
$
This has a quasi-periodic solution $\alpha_i = \lambda^i$, where $\lambda$ is such that $\lambda^{2k + 1} = w$. Computing $\tilde \D = \alpha  \D' \alpha^{-1}$ with such $\alpha$, we find that its coefficients $\tilde a_k$ are given by
 $$\tilde a_j = w  \lambda^{-j} a_j= \lambda^{2k + 1 - j} a_j$$ when $j$ is odd, and  $$\tilde a_j = \lambda^{-j} a_j$$ when $j$ is even. Upon a parameter change $s = \lambda^{-2} = w^{-\frac{2}{2k+1}}$, this coincides with formulas for the scaling given in \cite[Section 9]{khesin2013}.
%
\end{proof}
\begin{remark}
In \cite{khesin2013}, the invariance of short-diagonal maps under scaling was only established in dimensions $\leq 6$, while the general case was proved in~\cite{beffa2015}. With our definition, the invariance of pentagram maps under scaling is immediate.
\end{remark}
\begin{corollary}
For the classical, as well as short-diagonal and dented maps, first integrals obtained from our construction coincide with the ones obtained in \cite{ovsienko2010pentagram, khesin2013, khesin2016}.
\end{corollary}
\begin{proof}
Indeed, according to Corollary \ref{cor:int}, our integrals can be interpreted as spectral invariants of the monodromy for polygons obtained from the initial one by means of scaling $R_w$. But this is exactly the definition of first integrals in \cite{ovsienko2010pentagram, khesin2013, khesin2016}.
\end{proof}
\begin{remark}
For corrugated maps  of \cite{Gekhtman2016} our first integrals also coincide with the known ones. In fact, one can show more: for these maps, our refactorization description \eqref{refact} is equivalent to the one given in \cite[Proposition 4.10]{Gekhtman2016}. The refactorization description of \cite{Gekhtman2016} looks more complicated because it is given in terms of actual loop group elements (equivalently, pseudo-difference operators) $A_1(z)$, $A_2(z)$, as opposed to elements of the quotient by the $H \ttimes H$ action. Rewriting refactorization on the quotient as operator refactorization involves choosing a section of the action, which complicates the resulting formulas.
\end{remark}

\subsection{Poisson brackets for the short-diagonal map in 3D}\label{sec:poisson}
In this section we derive explicit formulas for Poisson brackets preserved by the short-diagonal pentagram map in 3D. The corresponding sets $J_\pm$ are $J_+ = \{-2,0,2\}$, $J_- = \{-1,1\}$ (the choice $J_+ = \{0,2,4\}$, $J_- = \{1,3\}$ indicated in Table \ref{tab:pentmaps} leads to the same map up to the shift of indices and hence gives rise to the same Poisson bracket). The phase space  of the associated map is the space of all twisted $n$-gons in $\RP^3$ modulo projective equivalence. We coordinatize that space as in  \cite[Section 5.2]{khesin2013}, namely we assign to a twisted $n$-gon $\{v_i \in \RP^3\}$ three periodic $n$-sequences $x_i, y_i, z_i$ defined as the following negative cross-ratios:
\begin{align*}
x_i &:= -[v_{i+4}, v_{i+5}, \langle v_i, v_{i+1}, v_{i+2} \rangle \cap \langle v_{i+4}, v_{i+5} \rangle, \langle v_{i+1}, v_{i+2}, v_{i+3} \rangle \cap \langle v_{i+4}, v_{i+5} \rangle],\\
&y_i := -[v_{i}, v_{i+1},  \langle v_{i+2} ,v_{i+3}, v_{i+4} \rangle \cap  \langle v_i, v_{i+1}\rangle,  \langle  v_{i+2} ,v_{i+4}, v_{i+5} \rangle \cap \langle v_i, v_{i+1}\rangle],\\
z_i &:=-[v_{i+4}, v_{i+5}, \langle v_i, v_{i+1}, v_{i+3} \rangle \cap \langle v_{i+4}, v_{i+5} \rangle, \langle v_{i+1}, v_{i+2}, v_{i+3} \rangle \cap \langle v_{i+4}, v_{i+5} \rangle].
\end{align*}
\begin{proposition} In these coordinates, the Poisson structure for the short-diagonal pentagram map in $\RP^3$ takes the following form:
\begin{multline}\label{sdb}
\begin{aligned}
\{ x_i, x_{i+1} \} = x_ix_{i+1}, \quad
&\{ x_i, x_{i+2} \} = x_ix_{i+2}w_{i+1},\quad \{ y_i, y_{i+2} \} = y_iy_{i+2}w_{i+1}, \quad \{z_i,z_{i+2}\} = z_iz_{i+2}w_i \\
&\{ x_i, y_{i-2} \} = x_iy_{i-2}w_{i-1},\quad
\{ x_i, y_{i+2} \} = -x_iy_{i+2}w_{i+1},\\
\{ x_i, z_{i-1} \} = x_i&z_{i-1}(w_{i-1}-1), \quad
\{ x_i, z_{i+1} \} = x_iz_{i+1}, \quad
\{ x_i, z_{i+3} \} = -x_iz_{i+3}w_{i+1},\\
\{ y_i, z_{i-1} \} = y_i&z_{i-1}(1-w_{i-1}), \quad
\{ y_i, z_{i+1} \} = -y_iz_{i+1}, \quad
\{ y_i, z_{i+3} \} = y_iz_{i+3}w_{i+1},
\end{aligned}
\end{multline}
where $w_i := y_{i+1}z_i$. 
\end{proposition} \begin{proof}
A direct computation shows that for any difference operator $\D = aT^{-2} + bT^{-1} + c + dT + eT^2$ representing the polygon $\{v_i\}$, the coordinates $x_i, y_i, z_i$ can be expressed in terms of coefficients of $\D$ as follows:
 \begin{align}\label{xyz}
x_{i-2} = -\frac{c_{i+1}e_i}{d_id_{i+1}}, \quad y_{i-2} = -\frac{a_{i+1}d_i}{b_i c_{i+1}}, \quad z_{i-2} = -\frac{b_{i+1}e_i}{c_i d_{i+1}} .
\end{align}
The Poisson bracket between coefficients of $\D$ is, by construction, the product bracket corresponding to the decomposition $\D = \D_+ + \D_-$, where $\D_+ = aT^{-2} + c + eT^2$, $\D_- = bT^{-1} + cT$. The bracket on operators $\D_+$ is defined as the restriction of the bracket $\pi^{(2)}$ on $2$-sparse operators, while the $\D_-$ part is endowed with the \textit{negative} of that bracket. Similarly to Example \ref{ex:sparsebracket}, we get
 \begin{align}
\begin{aligned}
 \{a_i,c_i&\} = \frac{1}{2}a_i c_i, \quad \{a_i,e_i\} = \frac{1}{2}a_i e_i, \quad \{c_i, e_i\} = \frac{1}{2} c_ie_i, \quad  \{c_i,a_{i+2}\} = \frac{1}{2}c_ia_{i+2},   \\  &\{c_i, c_{i+2}\} = a_{i+2}e_i, \quad \{e_i, c_{i+2}\} = \frac{1}{2} e_i c_{i+2}, \quad
 \{e_i,a_{i+4}\} = \frac{1}{2}e_ia_{i+4}
\end{aligned}
\end{align}
and
\begin{align}
   \{b_i,d_i\} = -\frac{1}{2}b_i d_i, \quad   \{d_i,b_{i+2}\} = -\frac{1}{2}b_ia_{i+2}.
\end{align}
It now remains to compute the brackets of functions \eqref{xyz} using these formulas. This is done by a straightforward calculation.
\end{proof}
\begin{remark}
As shown in \cite[Theorem 5.6]{khesin2013}, the short-diagonal map in $xyz$-coordinates reads
\begin{align*}
\tilde x_i = x_{i+1}\frac{\alpha_i}{\beta_i},\quad
\tilde y_i = \frac{x_{i-1}y_{i-2}z_i}{x_iz_{i-1}}\frac{\beta_{i-1}\beta_{i+2}}{\alpha_i\beta_{i+1}},\quad
\tilde z_i = \frac{x_{i+1}z_i}{x_i}\frac{\beta_{i-1}\beta_{i+2}}{\alpha_{i-1}\beta_{i}},
\end{align*}
where
$$
\alpha_i := 1 + y_{i-1} + z_{i+2} + y_{i-1}z_{i+2} - y_{i+1}z_i, \quad \beta_i := 1 + y_{i-1} + z_{i}. 
$$
It follows from our construction that this map preserves the above bracket. This can of course be verified with a computer algebra system.
\end{remark}

\subsection{Refactorization and Y-meshes}\label{sec:meshes}
In this section we outline the connection between the refactorization description of higher pentagram maps and the description in terms of Y-meshes given in \cite{glick2015}. Although we only consider the example of a short-diagonal pentagram map in $\RP^3$, it is quite likely that all the same arguments work for more general maps in any dimension. \par
Let us briefly recall the Y-mesh description of the short-diagonal map from \cite{glick2015}. A \textit{Y-pin} $S$ is four distinct points $S=\{a,b,c,d \in \Z^2\}$, satisfying certain technical conditions. Given a Y-pin  $S=\{a,b,c,d\}$, a \textit{Y-mesh of type $S$ and dimension $d$} is a map $v \colon \Z^2 \to \Proj^d$ such that the points $v(r + a)$, $v(r + b)$, $v(r + c)$, $v(r + d)$ are collinear for any $r \in \Z^2$. One can view any Y-mesh as a polygon depending on a discrete time variable $t \in \Z$. By definition, the $i$'th vertex of the polygon at time $t$ is given by $v(i,t)$. In what follows, we will only consider Y-meshes such that $v(i + n,t) = \phi(v(i ,t) )$ for a fixed projective transformation $\phi$. In other words, we assume that all the polygons defined by the Y-mesh are twisted $n$-gons with the same monodromy.\par
\begin{figure}[b]
\centering
\begin{tikzpicture}[, scale = 1]
    \foreach \x in {0,1,2,3,4}
    \foreach \y in {0,1,2,3,4}
    {
    \fill (\x,\y) circle (0.7pt);
    }
     \draw (2,3) circle (0.15);
         \draw (2,2) circle (0.15);
             \draw (1,1) circle (0.15);
                   \draw (3,1) circle (0.15);
\end{tikzpicture}
\caption{A Y-pin corresponding to the short-diagonal map in 3D.}\label{FigYP}
\end{figure}
The collinearity assumption on  $v(r + a)$, $v(r + b)$, $v(r + c)$, $v(r + d)$ defines a relation between the polygon $v(*,t)$ and the polygons corresponding to several previous time instances. Thus, Y-meshes can be regarded as dynamical systems. Since  the polygon $v(*,t)$ may be expressed in terms of polygons corresponding to several previous values of time, such a dynamical system is, generally speaking, defined on the space of $k$-tuples of polygons (as opposed to pentagram maps which are defined on polygons). Furthermore, those polygons need to satisfy certain additional restrictions. As an example, consider the Y-pin $S := \{(-1,0), (1,0), (0,1), (0,2)\}$ depicted in Figure~\ref{FigYP}. In this case, the horizontal level $v(*, t+2)$ may be expressed in terms of the previous two levels. Indeed, by definition of a Y-mesh, the vertex $v(i, t+2)$ may be reconstructed as the intersection of lines $\langle v(i-1, t), v(i+1, t)\rangle \cap \langle v(i-1, t+1), v(i+1, t+1)\rangle  $. Thus, in this case the Y-mesh may be viewed as a dynamical system on pairs of polygons. These polygons satisfy two additional conditions: \begin{itemize}\item The vertex $v(i,t+1)$ of the second polygon lies on the diagonal $\langle v(i-1, t), v(i+1, t)\rangle$ of the first polygon. \item The respective diagonals $\langle v(i-1, t), v(i+1, t)\rangle$ and $\langle v(i-1, t+1), v(i+1, t+1)\rangle  $ of the two polygons are coplanar. \end{itemize}
Further, the authors of \cite{glick2015} observe that in dimension $d =3 $ the square of the map $$(v(*, t), v(*, t+1)) \mapsto (v(*, t+1), v(*, t+2))$$ defined by the Y-pin depicted in Figure \ref{FigYP} is precisely the short-diagonal pentagram map. Indeed, we have $v(i-1, t+1) \in \langle v(i-2, t), v(i, t)\rangle$ and $v(i+1, t+1) \in \langle v(i, t), v(i + 2, t)\rangle$, so the point $v(i, t+2) \in \langle v(i-1, t+1), v(i+1, t+1)\rangle$ belongs to the plane $\langle v(i-2, t), v(i, t), v(i+2, t)\rangle$. Given also that $v(i, t+2) \in \langle v(i-1, t), v(i+1, t)\rangle$, we get
$$
v(i, t+2) \in \langle v(i-1, t), v(i+1, t)\rangle \cap \langle v(i-2, t), v(i, t), v(i+2, t)\rangle,
$$
which is precisely the definition of the short-diagonal map. Thus, the map defined by  the Y-pin depicted in Figure \ref{FigYP} can be viewed as the ``square root'' of the short-diagonal map. This square root, however, is not defined on the space of polygons itself, but on a certain extension of that space which consists of pairs of polygons satisfying two above-mentioned conditions. It can be shown, using purely geometric arguments, that this extended space is generically a finite cover of the space of polygons. In other words, given a level $v(*, t)$ of a Y-mesh of type depicted in Figure \ref{FigYP}, there are generically finitely many ways to reconstruct the next level $v(*, t+1)$ and thus all subsequent levels. Below we give an algebraic proof, by showing that this reconstruction problem is equivalent to a factorization problem for the difference operator $\D_+$ corresponding to the initial polygon $v(*, t)$. 

\par
Recall that the short-diagonal map in 3D corresponds to progressions $J_+ = \{-2,0,2\}$, $J_- = \{-1,1\}$. To every twisted $n$-gon in $\Proj^3$ we can assign two operators $\D_\pm \in \DO{n}{J_\pm}$ supported in those sets, which identifies the short-diagonal map with refactorization dynamics \eqref{mainRelation2}. Assume now that the polygon encoded by the operators $\D_\pm$ is realized as a level $v(*, t)$ of a Y-mesh of type depicted in Figure \ref{FigYP}. Let $V(i, t)$ be the lifts of points $v(i, t)$ to $\R^4$. Since the levels $v(*, t)$ and $v(*, t+2)$ are related by the short-diagonal map, their lifts $V(*, t)$, $V(*, t+2)$ may be chosen in such a way that $$\D_+ V(*, t) =-\D_- V(*, t) = V(*, t+2)$$ (cf. the proof of Theorem \ref{thm0}).  Furthermore, since $v(i, t+2) \in \langle v(i-1, t+1), v(i+1, t+2)\rangle$, there exists a difference operator $\D_+^{(1)}$ supported in $\{-1,1\}$ such that \begin{equation}V(*, t+2) = \D_+^{(1)} V(*, t+1).\end{equation} Analogously, there exists a difference operator $\D_+^{(2)}$ supported in $\{-1,1\}$ such that \begin{equation}\label{nextlevel}V(*, t+1) = \D_+^{(2)} V(*, t).\end{equation} Therefore, we have
\begin{equation}\label{factor}
(\D_+ - \D_+^{(1)}\D_+^{(2)}) V(*, t) = 0.
\end{equation}
But since both operators $\D_+$ and $\D_+^{(1)}\D_+^{(2)}$ and hence their difference are supported in $\{-2,0,2\}$, it follows that
$$
\D_+ = \D_+^{(1)}\D_+^{(2)}.
$$

Conversely, given such a factorization of $\D_+$, we can reconstruct the level $v(*, t+1)$ of the Y-mesh by using  \eqref{nextlevel}, and hence reconstruct all the subsequent levels. 

\begin{proposition}\label{factorProp}
A generic difference operator $\D$ supported in ${\{-2,0,2\}}$ has two factorizations of the form $\D = \D_1 \D_2$, where $\D_{i}$'s are supported in $\{-1,1\}$, if $n$ is odd, and four such factorizations if $n$ is even. Two factorizations $\D_1 \D_2$ and $\tilde \D_1 \tilde \D_2$ are considered the same if $\tilde \D_1 =  \D_1 \alpha^{-1}$ and $\tilde \D_2 =\alpha \D_2  $ for a certain $n$-periodic non-vanishing sequence $\alpha$.
\end{proposition}
\begin{remark}
The coefficients of the factors are, in general, complex numbers, even if the initial operator $\D$ is real.
\end{remark}
\begin{proof}[Proof of Proposition \ref{factorProp}]
The problem is equivalent to representing an operator   supported in ${\{0,2,4\}}$ as a product of two operators supported in $\{0,2\}$. If $n$ is odd, this problem further reduces, using the isomorphism described in Remark \ref{rem:cbi}, to representing an operator  $\D$ supported in ${\{0,1,2\}}$ as a product $\D_1\D_2$ of two operators supported in $\{0,1\}$. The latter problem has two different solutions for generic $\D$ since $\D_2$ is a right divisor of $\D$ if and only it annihilates a certain element of $\Ker \D$, and since $\D_2$ must be periodic, this element has to be of the two eigenvectors of the monodromy operator. Similarly, if $n$ is even, an operator supported in  ${\{0,2,4\}}$ can be identified with two $(n/2)$-periodic operators supported in ${\{0,1,2\}}$ (see Remark \ref{rem:cbi}), each of which has two different factorizations. Hence, in this case we generically have $2 \times 2 = 4$ distinct factorizations.
\end{proof}
Therefore, the square root of the short-diagonal map defined by the Y-pin depicted in Figure \ref{FigYP} acts on the space which is generically a $2$-to-$1$ or $4$-to-$1$ covering of the space of polygons. This space can be described as the space of triples of operators $\D_+^{(1)}, \D_+^{(2)}, \D_-$, all of which are supported in $\{-1,1\}$. These operators should be considered up to the action 
$$
\D_+^{(1)} \mapsto \alpha \D_+^{(1)} \beta^{-1}, \quad \D_+^{(2)} \mapsto \beta \D_+^{(2)} \gamma^{-1}, \quad \D_- \mapsto \alpha \D_- \gamma^{-1},
$$
where $\alpha, \beta, \gamma$ are $n$-quasi-periodic sequences with the same monodromy. This space projects to the space of polygons in $\Proj^3$ by means of the map
$$
\D_+^{(1)},  \D_+^{(2)},  \D_- \quad \longmapsto \quad  \D_+^{(1)}  \D_+^{(2)},  \D_-.
$$
Furthermore, the Y-mesh dynamics (i.e. the square root of the short-diagonal map) can be expressed in terms of difference operators as follows:
$$
\tilde \D_- \D_+^{(2)} = \tilde \D_+^{(1)}\D_-, \quad \tilde  \D_+^{(2)} =  \D_+^{(1)},
$$
which can also be described as the following refactorization:
$$
\D_-^{-1}  \D_+^{(1)}  \D_+^{(2)} \quad  \longmapsto \quad  \D_+^{(2)} \D_-^{-1}  \D_+^{(1)}.
$$
Since  $\tilde \D_+^{(2)} =  \D_+^{(1)}$, applying this refactorization twice we obtain the operator $\D_+^{(1)}\D_+^{(2)} \D_-^{-1} $, which is equivalent to the short-diagonal map. Thus, the Y-mesh interpretation of higher pentagram maps can be regarded as a step-by-step refactorization, where on each step one needs to solve a refactorization-type problem for \textit{binomial} difference operators (i.e. operators whose support consists of two elements). As shown in  \cite{glick2015}, each of these individual steps can be identified with a sequence of mutations in an appropriately defined cluster algebra. We conjecture that refactorization problems for binomial operators always admit a cluster description. An example of that is discussed in the next section. Namely, we show how the refactorization description of the classical pentagram map yields a description in terms of networks, in the spirit of \cite{Gekhtman2016}. Since network moves are well-known to correspond to cluster mutations, this also provides a cluster algebra description. 
\medskip
\subsection{From refactorization to networks}\label{sec:networks}

\begin{figure}[t]
\centering
\begin{tikzpicture}[rotate = 90, scale = 0.9]
\draw (0,-0.2) --  (0,5.2);
\draw (3,-0.2) -- (3,5.2);
\fill [opacity = 0.05] (0,-0.2) -- (3,-0.2) -- (3,5.2) -- (0,5.2) -- cycle;
\draw [->-] (0,2) -- (1,2);
\node at (0.5,2) [left] () {$d$};
\draw [->-] (1,2) -- (2,3);
\node at (1.5,2.5) [below] () {$c$};
\draw [->-] (2,1) -- (3,1);
\draw [->-] (1,2) -- (2,1);
\node at (1.5,1.5) [left] () {$e$};
\draw [->-] (2,3) -- (3,3);
\draw [->-] (1,4) -- (2,3);
\node at (1.5,3.5) [right] () {$b$};
\draw [->-] (0,4) -- (1,4);
\draw [->-] (1,4) -- (3,5);
\node at (2,4.5) [right] () {$a$};
\draw [->-] (0,0) -- (2,1);
\node at (1,0.5) [left] () {$f$};
\fill(0,0) circle (.4ex);
\fill(0,2) circle (.4ex);
\fill(0,4) circle (.4ex);
\fill(1,4) circle (.4ex);
\fill(2,1) circle (.4ex);
\fill(2,3) circle (.4ex);
\fill(3,5) circle (.4ex);
\fill(3,1) circle (.4ex);
\fill(3,3) circle (.4ex);
\fill(1,2) circle (.4ex);
\node at (0,0) [below] () {$1$};
\node at (0,2) [below] () {$0$};
\node at (0,4) [below] () {$-1$};
\node at (3,1) [above] () {$1$};
\node at (3,3) [above] () {$0$};
\node at (3,5) [above] () {$-1$};
%
%
%
%
%

\end{tikzpicture}
\caption{A network.}\label{FigNet}
\end{figure}

In this section we show how the refactorization approach to the classical pentagram map yields a description in terms of weighted directed networks, in the spirit of \cite{Gekhtman2016}. Such networks were introduced by A.\,Postnikov~\cite{postnikov2006total} to study totally positive Grassmannians. For the purposes of our paper, a \textit{network} is a directed graph embedded in an infinite strip, as shown in Figure \ref{FigNet}. All vertices located at one boundary component of the strip are $1$-valent sources labeled by integers (so there are countably many of them). Likewise, all vertices at the other boundary component are $1$-valent sinks also labeled by integers. All interior vertices are $3$-valent and are neither sources nor sinks. Some edges of the graph are assigned with numbers, called \textit{weights}. If no weight is explicitly assigned, it is assumed that the weight of the corresponding edge is $1$. We also assume for simplicity that there are no directed cycles.  \par
The \textit{weight of a directed path} in a network is the product of weights of edges on that path. The \textit{boundary measurement} between the source $i$ and sink $j$ is the sum of weights of all directed paths going from $i$ to $j$ (we will only consider networks for which every such sum is finite). The \textit{boundary measurement matrix} is the bi-infinite matrix whose entries are the boundary measurements (below we use the convention that the $(i,j)$ entry of that matrix corresponds to boundary measurement between the source $j$ and sink $i$). In what follows, we only consider networks whose boundary measurement matrices represent difference or pseudo-difference operators. If the boundary measurement matrix of a certain network represents an operator, we will also say that the network itself represents that operator.

\begin{figure}[b]
\centering
\begin{tikzpicture}[scale = 0.9]
\draw (-1,-0.2) --  (-1,5.2);
\draw (3,-0.2) -- (3,5.2);
\fill [opacity = 0.05] (-1,-0.2) -- (3,-0.2) -- (3,5.2) -- (-1,5.2) -- cycle;
\draw [->-] (0,2) -- (2,3);
\node at (1,2.5) [above] () {$a_1$};
\draw [->-] (0,4) -- (2,3);
\node at (1,3.5) [above] () {$b_1$};
\draw [->-] (0,4) -- (2,5);
\node at (1,4.5) [above] () {$a_2$};
\draw [->-] (0,0) -- (2,1);
\node at (1,0.5) [above] () {$a_{0}\,$};
\draw [->-] (0,2) -- (2,1);
\node at (1,1.5) [above] () {$b_{0}$};
\fill(0,0) circle (.4ex);
\fill(0,2) circle (.4ex);
\fill(0,4) circle (.4ex);
\fill(2,1) circle (.4ex);
\fill(2,3) circle (.4ex);
\fill(2,5) circle (.4ex);
\fill(-1,0) circle (.4ex);
\fill(-1,2) circle (.4ex);
\fill(-1,4) circle (.4ex);
\fill(3,1) circle (.4ex);
\fill(3,3) circle (.4ex);
\fill(3,5) circle (.4ex);
\draw [->-] (-1,0) -- (0,0);
\draw [->-] (-1,2) -- (0,2);
\draw [->-] (-1,4) -- (0,4);
\draw [->-] (2,1) -- (3,1);
\draw [->-] (2,3) -- (3,3);
\draw [->-] (2,5) -- (3,5);
\node at (-1,0) [left] () {$0$};
\node at (-1,2) [left] () {$1$};
\node at (-1,4) [left] () {$2$};
\node at (3,1) [right] () {$0$};
\node at (3,3) [right] () {$1$};
\node at (3,5) [right] () {$2$};
\draw (0,0) -- (0.4,-0.2);
\draw (2,5) -- (1.6,5.2);
%
%
%
%
%

\end{tikzpicture}
\qquad\qquad\qquad
\begin{tikzpicture}[scale = 0.9]
\draw (-1,-0.2) --  (-1,5.2);
\draw (3,-0.2) -- (3,5.2);
\fill [opacity = 0.05] (-1,-0.2) -- (3,-0.2) -- (3,5.2) -- (-1,5.2) -- cycle;
\draw [-<-] (0,1) -- (2,2);
\node at (1,1.5) [above] () {$-b_{0}\,\,$};
\draw [->-] (0,3) -- (2,2);
\node at (1,2.5) [above] () {$a_1^{-1}\,$};
\draw [-<-] (0,3) -- (2,4);
\node at (1,3.5) [above] () {$-b_1\,$};
\draw [->-] (0,1) -- (2,0);
\node at (1,0.5) [above] () {$a_{0}^{-1}\,$};
\draw [->-] (0,5) -- (2,4);
\node at (1,4.5) [above] () {$a_2^{-1}\,$};
\fill(0,1) circle (.4ex);
\fill(0,3) circle (.4ex);
\fill(0,5) circle (.4ex);
\fill(2,0) circle (.4ex);
\fill(2,2) circle (.4ex);
\fill(2,4) circle (.4ex);
\fill(-1,1) circle (.4ex);
\fill(-1,3) circle (.4ex);
\fill(-1,5) circle (.4ex);
\fill(3,0) circle (.4ex);
\fill(3,2) circle (.4ex);
\fill(3,4) circle (.4ex);
\draw [->-] (-1,1) -- (0,1);
\draw [->-] (-1,3) -- (0,3);
\draw [->-] (-1,5) -- (0,5);
\draw [->-] (2,0) -- (3,0);
\draw [->-] (2,2) -- (3,2);
\draw [->-] (2,4) -- (3,4);
\node at (-1,1) [left] () {$0$};
\node at (-1,3) [left] () {$1$};
\node at (-1,5) [left] () {$2$};
\node at (3,0) [right] () {$0$};
\node at (3,2) [right] () {$1$};
\node at (3,4) [right] () {$2$};
\draw (0,5) -- (0.4,5.2);
\draw (2,0) -- (1.6,-0.2);
%
%
%
%
%

\end{tikzpicture}

\caption{Networks representing the difference operator $a+bT$ and its inverse.}\label{FigDO}
\end{figure}

\begin{example} For two bi-infinite scalar sequences $a$, $b$, consider the difference operator $a+bT$. Figure~\ref{FigDO} shows networks representing that operator and its inverse (which is a pseudo-difference operator). To prove that these two networks represent inverse operators, one considers their concatenation, i.e. glues the sinks of one network to the sources of the other (which corresponds to composition of the corresponding operators), and shows that the resulting network represents the identity operator. Note that if the operator  $a+bT$ is periodic, then these networks are also periodic and can be thought of as networks on a cylinder, as in \cite{Gekhtman2016}.

\end{example}
Networks admit local transformations which do not change boundary measurements. These transformations are known as \textit{Postnikov moves}. Following \cite{Gekhtman2016}, we consider three types of moves depicted in Figure \ref{FigMove}. For the third move, the updated weights $\tilde w$, $\tilde x$, $\tilde y$, $\tilde z$ are rational functions of the initial weights $w$, $x$, $y$, $z$ whose particular form can be easily derived from preservation of boundary measurements and is irrelevant to our purposes. For other types of moves, weights do not change. \par
\begin{figure}
\centering
\begin{tabular}{>{\centering\arraybackslash} m{4cm} >{\centering\arraybackslash} m{1cm}>{\centering\arraybackslash} m{4cm} >{\centering\arraybackslash} m{3cm}}
\begin{tikzpicture}[, scale = 0.9, rotate = -90]

\draw [->-] (0,0) --  (1,0); 
\draw [->-] (-1,-1) --  (0,0); 
\draw [-<-] (-1,1)-- (0,0);
\draw [->-] (1,0)-- (2,-1);
\draw [-<-] (2,1) -- (1,0);
\node at (-0.5, -0.5) [above] () {$a$};
\node at (-0.5, 0.5) [above] () {$b$};
\node at (1.5, -0.5) [above] () {$c$};
\node at (1.5, 0.5) [above] () {$d$};
\fill(0,0) circle (.4ex);
\fill(1,0) circle (.4ex);
\end{tikzpicture} & $\longleftrightarrow$ & \begin{tikzpicture}[, scale = 0.9, rotate = -90]
\draw [-<-] (0,0) --  (0,-1); 
\draw [->-] (-1,-2) --  (0,-1); 
\draw [-<-] (-1,1)-- (0,0);
\draw [->-] (0,-1)-- (1,-2);
\draw [-<-] (1,1) -- (0,0);
\node at (-0.5, -1.5) [above] () {$a$};
\node at (-0.5, 0.5) [above] () {$b$};
\node at (0.5, -1.5) [above] () {$c$};
\node at (0.5, 0.5) [above] () {$d$};
\fill(0,0) circle (.4ex);
\fill(0,-1) circle (.4ex);
\end{tikzpicture}  & Type 1 \\ 
\begin{tikzpicture}[, scale = 0.9]

\draw [->-] (0,0) --  (1,0); 
\draw [->-] (-1,-1) --  (0,0); 
\draw [->-] (-1,1)-- (0,0);
\draw [->-] (1,0)-- (2,-1);
\draw [->-] (2,1) -- (1,0);
\node at (-0.5, -0.5) [above] () {$c\,$};
\node at (-0.5, 0.5) [above] () {$a$};
\node at (1.5, -0.5) [above] () {$d$};
\node at (1.5, 0.5) [above] () {$b$};

\fill(0,0) circle (.4ex);
\fill(1,0) circle (.4ex);
\end{tikzpicture} & $\longleftrightarrow$ & \begin{tikzpicture}[, scale = 0.9]
\draw [->-] (0,0) --  (0,-1); 
\draw [->-] (-1,-2) --  (0,-1); 
\draw [->-] (-1,1)-- (0,0);
\draw [->-] (0,-1)-- (1,-2);
\draw [->-] (1,1) -- (0,0);
\fill(0,0) circle (.4ex);
\fill(0,-1) circle (.4ex);
\node at (-0.5, -1.5) [above] () {$c\,$};
\node at (-0.5, 0.5) [above] () {$a$};
\node at (0.5, -1.5) [above] () {$d$};
\node at (0.5, 0.5) [above] () {$b$};
\end{tikzpicture}  & Type 2 \\ 
\begin{tikzpicture}[, scale = 0.9]
\draw [->-] (0,0) --  (1,1); 
\node at (0.5,0.5) [above] () {$x\,\,\,$};
\draw [->-] (1,1) --  (2,0); 
\node at (1.5,0.5) [above] () {$\,\,y$};
\draw [->-] (2,0)-- (1,-1);
\node at (1.5,-0.5) [below] () {$\,\,z$};
\draw [->-] (0,0)-- (1,-1);
\node at (0.5,-0.5) [below] () {$w\,\,\,$};
\draw [->-] (-1,0) -- (0,0);
\draw [->-] (2,0) -- (3,0);
\draw [->-] (1,2) -- (1,1);
\draw [->-] (1,-1) -- (1,-2);
\node at (-0.5, 0) [above] () {$a$};
\node at (1, 1.5) [right] () {$b$};
\node at (2.5, 0) [above] () {$c$};
\node at (1,-1.5) [right] () {$d$};
\fill(0,0) circle (.4ex);
\fill(1,1) circle (.4ex);
\fill(1,-1) circle (.4ex);
\fill(2,0) circle (.4ex);
\end{tikzpicture} & $\longleftrightarrow$ & \begin{tikzpicture}[, scale = 0.9]
\draw [-<-] (0,0) --  (1,1); 
\node at (0.5,0.5) [above] () {$\tilde x\,\,\,$};
\draw [->-] (1,1) --  (2,0); 
\node at (1.5,0.5) [above] () {$\,\,\tilde y$};
\draw [-<-] (2,0)-- (1,-1);
\node at (1.5,-0.5) [below] () {$\,\,\tilde z$};
\draw [->-] (0,0)-- (1,-1);
\node at (0.5,-0.5) [below] () {$\tilde w\,\,\,$};
\draw [->-] (-1,0) -- (0,0);
\node at (-0.5, 0) [above] () {$a$};
\node at (1, 1.5) [right] () {$b$};
\node at (2.5, 0) [above] () {$c$};
\node at (1,-1.5) [right] () {$d$};
\draw [->-] (2,0) -- (3,0);
\draw [->-] (1,2) -- (1,1);
\draw [->-] (1,-1) -- (1,-2);
\fill(0,0) circle (.4ex);
\fill(1,1) circle (.4ex);
\fill(1,-1) circle (.4ex);
\fill(2,0) circle (.4ex);
\end{tikzpicture}  & Type 3\end{tabular}
\caption{Postnikov moves.}\label{FigMove}
\end{figure}
\begin{figure}[h]
\centering

\begin{tikzcd}[ row sep = large, column sep = small]
\begin{tikzpicture}[]
\draw (0,-0.2) --  (0,5.2);
\clip (-1,-0.2) -- (6,-0.2) -- (6,5.2) -- (-1,5.2) -- cycle;
\fill [opacity = 0.05] (0,-0.2) -- (5,-0.2) -- (5,5.2) -- (0,5.2) -- cycle;
\draw [->-] (1,2) -- (2,3);
\draw [line width = 5, opacity = 0, shorten <= 20](1,2) -- (2,3);
\node at (1.5,2.5) [left] () {$c_1$};
\draw [line width = 5, opacity = 0, shorten <= 20](1,4) -- (2,3);
\draw [->-] (1,4) -- (2,3);
\node at (1.5,3.5) [left] () {$d_1$};
\draw [line width = 5, opacity = 0, shorten <= 20](1,4) -- (2,5);
\draw [->-] (1,4) -- (2,5);
\node at (1.5,4.5) [left] () {$c_2$};
\draw [line width = 5, opacity = 0, shorten <= 20](1,0) -- (2,1);
\draw [->-] (1,0) -- (2,1);
\node at (1.5,0.5) [left] () {$c_{0}$};
\draw [line width = 5, opacity = 0, shorten <= 20] (1,2) -- (2,1);;
\draw [->-] (1,2) -- (2,1);
\node at (1.5,1.5) [left] () {$d_{0}$};
\fill(0,0) circle (.4ex);
\fill(0,2) circle (.4ex);
\fill(0,4) circle (.4ex);
\draw [->-] (0,0) -- (1,0);
\draw [->-] (0,2) -- (1,2);
\draw [->-] (0,4) -- (1,4);
\fill(1,0) circle (.4ex);
\fill(1,2) circle (.4ex);
\fill(1,4) circle (.4ex);
\fill(2,1) circle (.4ex);
\fill(2,3) circle (.4ex);
\fill(2,5) circle (.4ex);
\draw [line width = 5, opacity = 0] (2,1) -- (3,1);
\draw [line width = 5, opacity = 0](2,3) -- (3,3);
\draw [line width = 5, opacity = 0] (2,5) -- (3,5);
\draw [->, line width = 1.5] (2,1) -- (3,1);
\draw [->, line width = 1.5] (2,3) -- (3,3);
\draw [->, line width = 1.5] (2,5) -- (3,5);
\node at (0,0) [left] () {$2$};
\node at (0,2) [left] () {$3$};
\node at (0,4) [left] () {$4$};
\draw (1,0) -- (1.2,-0.2);
\draw (2,5) -- (1.8,5.2);
\draw [line width = 5, opacity = 0] (2,5) -- (1.8,5.2);
\draw (5,-0.2) -- (5,5.2);
\draw [-<-] (3,1) -- (4,2);
\node at (3.5,1.5) [right] () {$-b_{0}\,\,$};
\draw [->-] (3,3) -- (4,2);
\node at (3.5,2.6) [right] () {$a_1^{-1}\,$};
\draw [-<-] (3,3) -- (4,4);
\node at (3.5,3.5) [right] () {$-b_1\,$};
\draw [->-] (3,1) -- (4,0);
\node at (3.5,0.6) [right] () {$a_{0}^{-1}\,$};
\draw [->-] (3,5) -- (4,4);
\node at (3.5,4.6) [right] () {$a_2^{-1}\,$};
\draw [line width = 5, opacity = 0, shorten >= 20](3,1) -- (4,2);
\draw [line width = 5, opacity = 0, shorten >= 20] (3,3) -- (4,2);
\draw [line width = 5, opacity = 0, shorten >= 20](3,3) -- (4,4);
\draw [line width = 5, opacity = 0, shorten >= 20](3,1) -- (4,0);
\draw [line width = 5, opacity = 0, shorten >= 20](3,5) -- (4,4);
\fill(3,1) circle (.4ex);
\fill(3,3) circle (.4ex);
\fill(4,0) circle (.4ex);
\fill(4,2) circle (.4ex);
\fill(4,4) circle (.4ex);
\fill(3,5) circle (.4ex);
\fill(5,0) circle (.4ex);
\fill(5,2) circle (.4ex);
\fill(5,4) circle (.4ex);
\draw [->-] (4,0) -- (5,0);
\draw [->-] (4,2) -- (5,2);
\draw [->-] (4,4) -- (5,4);
\node at (5,0) [right] () {$0$};
\node at (5,2) [right] () {$1$};
\node at (5,4) [right] () {$2$};
\draw (3,5) -- (3.2,5.2);
\draw (4,0) -- (3.8,-0.2);
\draw [line width = 5, opacity = 0](3,5) -- (3.2,5.2);
\end{tikzpicture} 
\arrow[]{rrrr}{\parbox{1.5cm}{\centering Type\,2 moves at all thick edges}} & & & & \begin{tikzpicture}[]
\clip (-1,-0.2) -- (6,-0.2) -- (6,5.2) -- (-1,5.2) -- cycle;
\draw (0,-0.2) --  (0,5.2);
\fill [opacity = 0.05] (0,-0.2) -- (5,-0.2) -- (5,5.2) -- (0,5.2) -- cycle;
\draw [->-] (1,2) -- (2.5,2.75);
\draw [->-] (1,4) -- (2.5,3.25);
\draw [->-] (1,4) -- (2.5,4.75);

\fill [opacity = 0.1] (2.5,4.75) -- (4,4) -- (2.5,3.25) -- (1,4) -- cycle;
\fill [opacity = 0.1] (2.5,2.75) -- (4,2) -- (2.5,1.25) -- (1,2) -- cycle;
\fill [opacity = 0.1] (2.5,0.75) -- (4,0) -- (2.5,-0.75) -- (1,0) -- cycle;
\draw [->-] (1,0) -- (2.5,0.75);
\draw [->-] (1,2) -- (2.5,1.25);
\fill(0,0) circle (.4ex);
\fill(0,2) circle (.4ex);
\fill(0,4) circle (.4ex);
\draw [->-] (0,0) -- (1,0);
\draw [->-] (0,2) -- (1,2);
\draw [->-] (0,4) -- (1,4);
\fill(1,0) circle (.4ex);
\fill(1,2) circle (.4ex);
\fill(1,4) circle (.4ex);
\fill(2.5,0.75) circle (.4ex);
\fill(2.5,3.25) circle (.4ex);
\draw [->-] (2.5,1.25) -- (2.5,0.75);
\draw [->-] (2.5,3.25) -- (2.5,2.75);
\draw [->-] (2.5,5.2) -- (2.5,4.75);
\node at (0,0) [left] () {$2$};
\node at (0,2) [left] () {$3$};
\node at (0,4) [left] () {$4$};
\draw (1,0) -- (1.4,-0.2);
\draw (5,-0.2) -- (5,5.2);
\draw [-<-] (2.5,1.25) -- (4,2);
\draw [->-] (2.5,2.75) -- (4,2);
\draw [-<-] (2.5,3.25) -- (4,4);
\draw [->-] (2.5,0.75) -- (4,0);
\draw [->-] (2.5,4.75) -- (4,4);
\fill(2.5,1.25) circle (.4ex);
\fill(2.5,2.75) circle (.4ex);
\fill(4,0) circle (.4ex);
\fill(4,2) circle (.4ex);
\fill(4,4) circle (.4ex);
\fill(2.5,4.75) circle (.4ex);
\fill(5,0) circle (.4ex);
\fill(5,2) circle (.4ex);
\fill(5,4) circle (.4ex);
\draw [->-] (4,0) -- (5,0);
\draw [->-] (4,2) -- (5,2);
\draw [->-] (4,4) -- (5,4);
\node at (5,0) [right] () {$0$};
\node at (5,2) [right] () {$1$};
\node at (5,4) [right] () {$2$};
\draw (4,0) -- (3.6,-0.2);

\end{tikzpicture} \arrow[]{dd}{\parbox{3cm}{\centering Type\,3 moves at all shaded faces}}\\
\\
\begin{tikzpicture}[]
\draw (0,-0.2) --  (0,5.2);
\clip (-1,-0.2) -- (6,-0.2) -- (6,5.2) -- (-1,5.2) -- cycle;
\fill [opacity = 0.05] (0,-0.2) -- (5,-0.2) -- (5,5.2) -- (0,5.2) -- cycle;
\draw [-<-] (1,2) -- (2,3);
\draw [line width = 5, opacity = 0, shorten <= 20](1,2) -- (2,3);

\draw [line width = 5, opacity = 0, shorten <= 20](1,4) -- (2,3);
\draw [->-] (1,4) -- (2,3);
\draw [line width = 5, opacity = 0, shorten <= 20](1,4) -- (2,5);
\draw [-<-] (1,4) -- (2,5);
\draw [line width = 5, opacity = 0, shorten <= 20](1,0) -- (2,1);
\draw [-<-] (1,0) -- (2,1);
\draw [line width = 5, opacity = 0, shorten <= 20] (1,2) -- (2,1);;
\draw [->-] (1,2) -- (2,1);
\fill(0,0) circle (.4ex);
\fill(0,2) circle (.4ex);
\fill(0,4) circle (.4ex);
\draw [->-] (0,0) -- (1,0);
\draw [->-] (0,2) -- (1,2);
\draw [->-] (0,4) -- (1,4);
\fill(1,0) circle (.4ex);
\fill(1,2) circle (.4ex);
\fill(1,4) circle (.4ex);
\fill(2,1) circle (.4ex);
\fill(2,3) circle (.4ex);
\fill(2,5) circle (.4ex);
\draw [->-] (2,1) -- (3,1);
\draw [->-] (2,3) -- (3,3);
\draw [->-] (2,5) -- (3,5);
\node at (0,0) [left] () {$2$};
\node at (0,2) [left] () {$3$};
\node at (0,4) [left] () {$4$};
\draw [dashed] (2.5,-0.2) -- (2.5,5.2);
\draw (1,0) -- (1.2,-0.2);
\draw (2,5) -- (1.8,5.2);
\draw [line width = 5, opacity = 0] (2,5) -- (1.8,5.2);
\draw (5,-0.2) -- (5,5.2);
\draw [->-] (3,1) -- (4,2);
\draw [->-] (3,3) -- (4,2);
\draw [->-] (3,3) -- (4,4);
\draw [->-] (3,1) -- (4,0);
\draw [->-] (3,5) -- (4,4);

\fill(3,1) circle (.4ex);
\fill(3,3) circle (.4ex);
\fill(4,0) circle (.4ex);
\fill(4,2) circle (.4ex);
\fill(4,4) circle (.4ex);
\fill(3,5) circle (.4ex);
\fill(5,0) circle (.4ex);
\fill(5,2) circle (.4ex);
\fill(5,4) circle (.4ex);
\draw [->-] (4,0) -- (5,0);
\draw [->-] (4,2) -- (5,2);
\draw [->-] (4,4) -- (5,4);
\node at (5,0) [right] () {$0$};
\node at (5,2) [right] () {$1$};
\node at (5,4) [right] () {$2$};
\node at (2.3,1) [below] () {$3$};
\node at (2.3,3) [below] () {$4$};
\node at (2.3,5) [below] () {$5$};

\draw (3,5) -- (3.2,5.2);
\draw (4,0) -- (3.8,-0.2);

\draw [line width = 5, opacity = 0](3,5) -- (3.2,5.2);


\end{tikzpicture}  & & & & \begin{tikzpicture}[]
\clip (-1,-0.2) -- (6,-0.2) -- (6,5.2) -- (-1,5.2) -- cycle;
\draw (0,-0.2) --  (0,5.2);
\fill [opacity = 0.05] (0,-0.2) -- (5,-0.2) -- (5,5.2) -- (0,5.2) -- cycle;
\draw [-<-] (1,2) -- (2.5,2.75);
\draw [->-] (1,4) -- (2.5,3.25);
\draw [-<-] (1,4) -- (2.5,4.75);

\draw [-<-] (1,0) -- (2.5,0.75);
\draw [->-] (1,2) -- (2.5,1.25);
\fill(0,0) circle (.4ex);
\fill(0,2) circle (.4ex);
\fill(0,4) circle (.4ex);
\draw [->-] (0,0) -- (1,0);
\draw [->-] (0,2) -- (1,2);
\draw [->-] (0,4) -- (1,4);
\fill(1,0) circle (.4ex);
\fill(1,2) circle (.4ex);
\fill(1,4) circle (.4ex);
\fill(2.5,0.75) circle (.4ex);
\fill(2.5,3.25) circle (.4ex);
\draw [->, line width = 1.5] (2.5,1.25) -- (2.5,0.75);
\draw [->,  line width = 1.5] (2.5,3.25) -- (2.5,2.75);
\draw [->,  line width = 1.5] (2.5,5.2) -- (2.5,4.75);
\node at (0,0) [left] () {$2$};
\node at (0,2) [left] () {$3$};
\node at (0,4) [left] () {$4$};
\draw (1,0) -- (1.4,-0.2);
\draw (5,-0.2) -- (5,5.2);
\draw [->-] (2.5,1.25) -- (4,2);
\draw [->-] (2.5,2.75) -- (4,2);
\draw [->-] (2.5,3.25) -- (4,4);
\draw [->-] (2.5,0.75) -- (4,0);
\draw [->-] (2.5,4.75) -- (4,4);
\fill(2.5,1.25) circle (.4ex);
\fill(2.5,2.75) circle (.4ex);
\fill(4,0) circle (.4ex);
\fill(4,2) circle (.4ex);
\fill(4,4) circle (.4ex);
\fill(2.5,4.75) circle (.4ex);
\fill(5,0) circle (.4ex);
\fill(5,2) circle (.4ex);
\fill(5,4) circle (.4ex);
\draw [->-] (4,0) -- (5,0);
\draw [->-] (4,2) -- (5,2);
\draw [->-] (4,4) -- (5,4);
\node at (5,0) [right] () {$0$};
\node at (5,2) [right] () {$1$};
\node at (5,4) [right] () {$2$};
\draw (4,0) -- (3.6,-0.2);

\end{tikzpicture}
\arrow[swap]{llll}{\parbox{1.5cm}{\centering Type\,1 moves at all thick edges}} 
\end{tikzcd}.

\caption{A network representing a pseudo-difference operator $  (a + bT)^{-1}(cT^2 + dT^3)$ and its refactorization.}\label{FigPSIDO}
\end{figure}
We now show how to use Postnikov moves to encode refactorization of pseudo-difference operators. We will do that using the classical pentagram map as an example. Consider the progressions $J_+ = \{0,1\}$, $J_- = \{2,3\}$. Then the equation 
$
\tilde{\D}_+  {\D}_- = \tilde {\D}_-  {\D}_+,
$ 
where the operators $\D_\pm$ and $\tilde \D_\pm$ are supported in $J_\pm$, 
encodes the inverse pentagram map. Accordingly, the pentagram map itself can be described by
$
{\D}_+  \tilde{\D}_- =  {\D}_-  \tilde {\D}_+,
$
which is the same as
\begin{equation}\label{eq:pentnet}
\tilde{\D}_-  \tilde {\D}_+^{-1} = \D_+^{-1}\D_-.
\end{equation}
Thus, an application of the pentagram map can be thought of as rewriting an operator of the form $ \D_+^{-1}\D_-$ as~$\tilde{\D}_-  \tilde {\D}_+^{-1}$. This operation can be represented as a sequence of Postnikov moves, as follows. The network representing  $ \D_+^{-1}\D_-$, where $\D_+ = a + bT$ and $\D_- = cT^2 + dT^3$ is basically the concatenation of networks in Figure \ref{FigDO}, up to a change of weights and shift of indices, see upper left picture in Figure~\ref{FigPSIDO}. Applying Postnikov moves as shown in the figure (the figure does not show transformations of weights since those are irrelevant) results in the network depicted in the bottom left picture. That resulting network represents an operator of the form~$\tilde{\D}_-  \tilde {\D}_+^{-1}$, as can be seen by cutting it along the dashed line and labeling the newly obtained boundary vertices as shown (simply put, the left half of the new network looks the same as the right half of the initial one, and vice versa). Furthermore, since this new network is obtained from the initial one by Postnikov moves, these networks represent the same operator:
$$
\tilde{\D}_-  \tilde {\D}_+^{-1} = \D_+^{-1}\D_-,
$$
as required. Thus, the pentagram map can be represented as a sequence of Postnikov moves. Furthermore, it is well known that Postnikov moves give rise to cluster transformations of certain variables associated with faces, which gives the cluster description of the pentagram map, see \cite{Gekhtman2016}.
\begin{remark}
In \cite{Gekhtman2016}, the authors consider two different networks describing the pentagram map, in a sense dual to each other. One of their networks coincides with the one shown in the upper right picture in Figure \ref{FigPSIDO}, cf. \cite[Figure 14]{Gekhtman2016}. Thus, their network is obtained from ours by type 2 Postnikov moves. The advantage of our approach is that we obtain networks directly from the refactorization description and hence essentially from the geometry of the map, while in  \cite{Gekhtman2016} the identification between maps and networks is done at the level of formulas.
\end{remark}

More generally, one gets a network description for all refactorization corresponding to $J_\pm$ of the form $\{k, k+1\}$, thus recovering the results of \cite{Gekhtman2016}. It is an open problem whether it is possible to represent other pentagram maps using networks. This problem reduces to the question of constructing networks representing operators with support other than  $\{k, k+1\}$. This can definitely be done by means of factorization, as in the previous section. However, the weights of so obtained networks will not be rational functions in terms of the initial data. It is an interesting question whether one can represent an operator supported, say, in $\{0,1,2\}$ by means of a network whose weights are rational in terms of the operator coefficients. If this can be done, one may hope to obtain a cluster description of higher pentagram maps.

\medskip
\section{Open problems}\label{sec:problems}

\textbf{1. Relation to cluster algebras.} The classical pentagram map, as well as pentagram maps on corrugated polygons, can be described as sequences of cluster mutations~\cite{GLICK20111019, Gekhtman2016}. It would be interesting to find a similar description for more general pentagram maps on $J$-corrugated polygons or, even more generally, the maps $\Psi_{J_\pm}$ associated with arbitrary pairs of progressions with the same common difference. \par
Short-diagonal and dented maps were recently treated from the cluster perspective in \cite{glick2015} (see also  Section \ref{sec:meshes} above), where the authors introduced certain variables which transform, under the corresponding pentagram map, according to a cluster rule. However, the definition of those variables involves introduction of the $k$'th root of the corresponding map, which in general results in multivalued functions on the space of polygons (as we show in Section \ref{sec:meshes}, computation of such a root is equivalent to a factorization problem for a certain difference operator; in general, this operation cannot be performed using only rational functions). Do there exist single-valued cluster variables for short-diagonal, dented, and more general maps studied in the present paper? A possible approach to this problem is outlined in Section \ref{sec:networks}: first construct networks representing arbitrary difference operators and their inverses, and then show that refactorization is equivalent to a sequence of Postnikov moves.
\par
A related question is whether our maps fit into a construction of \cite{Goncharov2013} of integrable systems associated with dimer models on bipartite graphs, or perhaps some generalized version of it.
\\
\\
\textbf{2. Refactorization and Y-meshes.} Generalize the approach of Section~\ref{sec:meshes}  to all types of Y-meshes. What is the precise relation between maps described in the present paper and maps that admit a Y-mesh description? In particular, is it possible to interpret the cluster dynamics of \cite{glick2015} as refactorization of ratios of binomial difference operators, as in Section~\ref{sec:meshes} above?
\\
\\
\textbf{3. Maps associated with pairs of non-disjoint progressions.} In this paper we constructed refactorization maps associated with pairs of progressions $J_\pm \subset \Z$ with the same common difference. When these progressions are disjoint, such maps can be interpreted as pentagram-type maps. What is a geometric interpretation in the non-disjoint case?
\\
\\
\textbf{4. The leapfrog map.} Give a geometric proof of the fact that for $J_+ = \{-1,0\}$, $J_- = \{0,1\}$ our construction leads to the leapfrog map of \cite{Gekhtman2016} (cf. Remark \ref{rm:leapfrog}).
\\
\\
\textbf{5. Integrability.} For all maps $\Psi_{J_\pm}$ associated with pairs of progressions we constructed a Lax representation with spectral parameter and a Poisson structure such that the first integrals coming from the Lax representation Poisson-commute. This suggests that all these maps are both algebraically and Liouville integrable. Find a proof of this fact, i.e. show that the joint levels sets of first integrals are Lagrangian submanifolds of symplectic leaves, that each of those submanifolds can be identified with an open subset in the Jacobian of the corresponding spectral curve, and that a suitable power of the map  $\Psi_{J_\pm}$  is a translation relative to the natural group structure on the Jacobian.
\\
\\
\textbf{6. Difference operators with matrix coefficients and pentagram maps on Grassmannians.} The construction of the present paper can be generalized to difference operators with matrix coefficients. Does this lead to pentagram maps on Grassmannians defined in \cite{felipe2015}? How are the corresponding Poisson structures related to double brackets of \cite{ovenhouse2018non}?
\\
\\
\textbf{7. Partial difference operators and the Laplace transform.} One can generalize the construction of the present paper to partial difference operators supported in arithmetic progressions $J_\pm \subset \Z^2$.  This leads to pentagram-type maps defined on polyhedra. The simplest example of such a map is the discrete Laplace transform of \cite{doliwa1997geometric} corresponding to $J_+= \{(0,0), (1,0) \}$, $J_- = \{(0,1), (1,1) \}$. Are maps of this type integrable? \par
Note that pentagram map as well as its generalizations to corrugated polygons can be thought of as reductions of the Laplace transform, see \cite{affolter2019vector}. This should correspond to certain reductions of partial difference operators to ordinary ones.
\\
\\
\textbf{8. Poisson structures on reductions of difference operators.} 
Poisson structures studied in the present paper arise as reductions of structures on \textit{rational pseudo-difference} operators. One can also study Poisson structures on polygons arising as reductions of \textit{difference} operators,  see Remark~\ref{ratvsint}. For example, taking $d = 1$ and coordinatizing the moduli space of polygons in $\RP^1$ by means of cross-ratios of quadruples of consecutive vertices, one gets the following Poisson bracket:
\begin{equation}
\{x_i, x_{i+1}\} = x_i x_{i+1}(x_i + x_{i+1} - 1), \quad \{x_i, x_{i+2}\} = x_i x_{i+1}x_{i+2}.
\end{equation}
This bracket is well-known in relation to the Volterra lattice and also arises in the study of {cross-ratio dynamics} on polygons \cite{arnold2018cross, suris1999integrable}. Furthermore, this structure is often considered as a lattice analogue of the Virasoro algebra \cite{faddeev2016liouville}.
Similarly, computing the bracket on polygons in $\RP^2$, one recovers the Belov-Chaltikian lattice $W_3$-algebra \cite{belov1993lattice}. More generally, we believe that Poisson structures on polygons obtained by reduction from difference operators can be viewed as lattice versions of classical $W$-algebras. In particular, we conjecture that these structures coincide with the ones constructed by means of difference Drinfeld-Sokolov reduction \cite{beffa2013hamiltonian}. One interesting property that such structures have is that, in contrast to Poisson brackets studied in the present paper, they restrict to the space of closed polygons.

\newcounter{ai}
\newcommand{\ai}[1]
{\stepcounter{ai}$^{\bf AI\theai}$%
\footnotetext{\hspace{-3.7mm}$^{\blacksquare\!\blacksquare}$
{\bf AI\theai:~}#1}}

\newcounter{bk}
\newcommand{\bk}[1]
{\stepcounter{bk}$^{\bf BK\thebk}$%
\footnotetext{\hspace{-3.7mm}$^{\blacksquare\!\blacksquare}$
{\bf BK\thebk:~}#1}}

\par\medskip

\bibliographystyle{plain}
\bibliography{pent.bib}

\end{document}